\documentclass[12pt]{article}

\usepackage[utf8]{inputenc}
\usepackage[english]{babel}

\usepackage[a4paper, total={6.3in, 8.9in}]{geometry}

\usepackage{verbatim}
\usepackage{enumitem}
\usepackage{array,multirow}
\usepackage{enumitem}
\usepackage{amssymb,amsfonts,amsmath,amsthm}
\usepackage{url}
\usepackage{tikz}
\usepackage{bbm}

\usetikzlibrary{arrows.meta}
\usetikzlibrary{calc}

\usepackage{xcolor}
\usepackage{ctable}
\usepackage[colorlinks=true,linkcolor=blue,citecolor=violet]{hyperref}
\usepackage[nameinlink, noabbrev, capitalize]{cleveref}
\usepackage{aliascnt}
\usepackage{tikz}
\usetikzlibrary{matrix,arrows}
\usetikzlibrary{decorations.markings}

\newtheorem{theorem}{Theorem}[section]

\newtheorem{corollary}[theorem]{Corollary}
\newtheorem{proposition}[theorem]{Proposition}
\newtheorem{claim}[theorem]{Claim}

\newtheorem{remark}[theorem]{Remark}
\newtheorem{openproblem}[theorem]{Open Problem}

\crefname{theorem}{Theorem}{Theorems}
\Crefname{theorem}{Theorem}{Theorems}
\crefname{corollary}{Corollary}{Corollaries}
\Crefname{corollary}{Corollary}{Corollaries}
\crefname{lemma}{Lemma}{Lemmas}
\Crefname{lemma}{Lemma}{Lemmas}
\crefname{proposition}{Proposition}{Propositions}
\Crefname{proposition}{Proposition}{Propositions}
\crefname{claim}{Claim}{Claims}
\Crefname{claim}{Claim}{Claims}
\crefname{inq}{Inequality}{inequalities}
\crefname{inq}{Inequality}{inequalities}

\newaliascnt{inq}{equation}
\aliascntresetthe{inq}
\creflabelformat{inq}{#2\textup{(#1)}#3}
\makeatletter

\def\endineq{\eqno \hbox{\@eqnnum}$$\@ignoretrue}
\makeatother

\numberwithin{inq}{section} 

\numberwithin{equation}{section}

\newcommand{\support}{\mathsf{supp}}

\newcommand{\deff}{\triangleq}

\newcommand{\eps}{\varepsilon}

\newcommand{\p}{{\mathfrak{p}}}
\newcommand{\q}{{\mathfrak{q}}}
\newcommand{\pp}{{\mathfrak{P}}}
\newcommand{\qq}{{\mathfrak{Q}}}
\newcommand{\FF}{{\mathbb{F}}}
\newcommand{\ZZ}{{\mathbb{Z}}}
\newcommand{\trcode}{\mathsf{TC}}
\newcommand{\tr}{{\mathsf{Tr}}}
\newcommand{\rrs}{{\mathcal{L}}}
\renewcommand{\v}{{\upsilon}}
\newcommand{\wrs}{\mathsf{WS}}
\newcommand{\spn}{{\mathsf{Span}}}
\newcommand{\con}{{\mathrm{Con}}}

\newcommand{\vr}{\mathcal{O}}
\newcommand{\vrc}{\mathcal{O}'}

\begin{document}
\pagenumbering{gobble}

\title{Tracing AG Codes: Toward Meeting the \\Gilbert--Varshamov Bound}

\date{}

\author{
  Gil Cohen\thanks{Tel Aviv University. \texttt{gil@tauex.tau.ac.il}. Supported by ERC starting grant 949499 and by the Israel Science Foundation grant 2989/24.  	
  }\and
  Dean Doron\thanks{Ben Gurion University. \texttt{deand@bgu.ac.il}. Supported in part by NSF-BSF grant 2022644.}
  \and
  Noam Goldgraber\thanks{Ben Gurion University and Tel Aviv University. \texttt{goldgrab@post.bgu.ac.il}. Supported by NSF-BSF grant 2022644.}
  \and
  Tomer Manket\thanks{Tel Aviv University. \texttt{tomermanket@mail.tau.ac.il}. Supported by ERC starting grant 949499.}
}

\maketitle

\begin{abstract}

One of the oldest problems in coding theory is to match the Gilbert--Varshamov bound with explicit binary codes. Over larger---yet still constant-sized---fields, algebraic-geometry codes are known to beat the GV bound. In this work, we leverage this phenomenon by taking traces of AG codes. Our hope is that the margin by which AG codes exceed the GV bound will withstand the parameter loss incurred by taking the trace from a constant field extension to the binary field. In contrast to concatenation, the usual alphabet-reduction method, our analysis of trace-of-AG (TAG) codes uses the AG codes’ algebraic structure throughout -- including in the alphabet-reduction step.

Our main technical contribution is a Hasse--Weil–type theorem that is well-suited for the analysis of TAG codes. The classical theorem (and its Grothendieck trace-formula extension) are inadequate in this setting. Although we do not obtain improved constructions, we show that a constant-factor strengthening of our bound would suffice. We also analyze the limitations of TAG codes under our bound and prove that, in the high-distance regime, they are inferior to code concatenation. Our Hasse--Weil–type theorem holds in far greater generality than is needed for analyzing TAG codes. In particular, we derive new estimates for exponential sums.
	
\end{abstract}

\newpage
\setcounter{tocdepth}{2}
\tableofcontents
\newpage  
\pagenumbering{arabic}

\section{Introduction}\label{sec:introduction}

The study of the tradeoff between rate and distance is as old as coding theory itself, dating back to the late 1940s and early 1950s. Gilbert \cite{Gilbert1952} proved the existence of codes with distance~$\delta$ and rate
$
\rho \ge 1 - H_2(\delta),
$
and Varshamov \cite{Varshamov1957} subsequently proved that such codes can be taken to be linear. Naturally, the ultimate goal is to achieve \emph{explicit} constructions, and indeed much effort has focused on obtaining efficiently encodable codes with a good rate-vs.-distance tradeoff.

For binary codes, notable early constructions include Justesen's code \cite{Jus72}, the first explicit construction to get constant rate and constant relative distance, and expander-based codes \cite{Tan82,SS96}, which further enjoy very efficient decoding. The large distance regime, $\delta = 1/2-\eps$, has been the subject of extensive and fruitful research over
the past decades, in part due to its importance in complexity theory via its relation to small-bias sets. The GV bound guarantees the existence of codes with such distance $\delta$ and rate $\Omega(\eps^{2})$. Earlier explicit constructions achieved suboptimal rates~\cite{AGHP,NN93,Ben-Aroya_Ta-Shma}, but a breakthrough construction of Ta\hbox{-}Shma attained nearly optimal rate $\eps^{2+o(1)}$ via a sophisticated expander-based bias-reduction technique \cite{T17}. Progress toward the GV bound in the regime where \(\delta\) is bounded away from \(\tfrac12\) has been comparatively limited.

Constructing codes over large alphabets has proved easier, particularly with respect to the rate–distance tradeoff. Reed--Solomon codes are the prototypical example: they attain the optimal tradeoff by meeting the Singleton bound with alphabet size equal to the block length. Consequently, a particularly useful tool for constructing codes over small alphabets is \emph{code concatenation}, which reduces the alphabet size by combining large-alphabet outer codes with short, small-alphabet inner codes, yielding arbitrarily long codes over a small alphabet. Importantly, in the $\delta = 1/2-\eps$ regime, several constructions are concatenation-based (say, \cite{AGHP,Ben-Aroya_Ta-Shma}), and in particular, one can show that concatenating optimal algebraic-geometric codes with the Hadamard code yield codes of rate $\Omega(\eps^3)$. 

Remarkably, the GV bound is suboptimal for sufficiently large constant alphabets. A major line of research, initiated by Goppa's 1981 introduction of algebraic-geometry (AG) codes \cite{Goppa1981Doklady}, constructs codes from algebraic curves over finite fields. This approach laid the foundation for the breakthrough of Tsfasman, Vl\u{a}du\c{t}, and Zink \cite{TsfasmanVladutZink1982}, who proved that for sufficiently large fields (e.g., already over $\mathbb{F}_{49}$), AG codes can asymptotically achieve a rate–distance tradeoff exceeding the GV bound. Their result relies on families of curves with many rational points relative to their genus. Subsequently, Garcia and Stichtenoth \cite{GarciaStichtenoth1995,GarciaStichtenoth1996} constructed explicit recursive towers of function fields attaining the Drinfel’d-Vl\u{a}du\c{t} \cite{DrinfeldVladut1983FAIA} upper bound on the number of rational points, thereby providing explicit families of AG codes that meet the Tsfasman–Vl\u{a}du\c{t}-Zink bound in a fully constructive way. Moreover, as the field size increases, the quantitative improvement becomes more pronounced, and the range of parameters for which the GV bound can be exceeded broadens accordingly.

\subsection{Our Approach: Trace-Based Alphabet Reduction for AG Codes}

Inspired by the work of Kopparty, Ta-Shma, and Yakirevitch~\cite{KTY24} and by a talk of Ta-Shma at the Simons Institute (Berkeley)~\cite{TaShma2024AGcodes}, this paper studies a family of explicit constructions over constant-size fields that aim to meet---or even surpass---the GV bound, with a particular focus on binary codes. Our primary technical contribution is to develop tools for analyzing these candidate constructions and for understanding their limitations. The resulting statements are general and extend beyond the original motivation, with potential further applications.

Our idea is simple: we aim to leverage the underlying algebraic structure of AG codes---which enables them to beat the GV bound---in the alphabet-reduction step as well. The hope is that, by exploiting this structure, the reduction will not substantially compromise the excellent rate–distance tradeoff of the original AG code. In its most basic form, the alphabet-reduction method we propose applies the field trace to each coordinate of every codeword. While this is the variant we focus on in this work, we view it as a special case within a broader family of constructions whose common feature is the aforementioned strategy of leveraging structure for alphabet reduction.

This approach is, in a sense, an antithesis of the off-the-shelf technique of the black-box analysis of code concatenation. The latter ignores the code’s internal structure: the rate and distance of the resulting code are simply the products of the corresponding parameters of the outer and inner codes (see \cref{sec:preliminaries}). 
One can hope to get better guarantees by exploiting additional, more complicated, structure.\footnote{We note that while our trace code approach is inherently ``non black-box'', there have been few attempts at exploiting structure in concatenation-based construction, most recently in \cite{DMW24}, with a similar goal of attaining the GV bound.}

Concatenation is reminiscent of other composition-type primitives, such as the zig–zag product in expander-graph constructions, where the roles of rate and distance are played by the graph’s degree and spectral expansion. In contrast, exploiting more of the underlying structure—in the graph setting, the entire spectrum—yields stronger, and in fact optimal, analyses \cite{CCM24}.

Of course, this is not the first work to consider trace codes: they have been studied for decades, most notably since Delsarte \cite{Del75}, who established their connection to subfield subcodes and, in particular, to dual BCH codes. Trace codes of AG codes have likewise been investigated since the 1990s (e.g., \cite{Vlugt,skorobogatov1991parameters,trace_code_dimension} and subsequent works). However, the aspects examined in that literature, largely motivated by the analysis of BCH and cyclic codes, are of limited relevance to our goals here. The papers by
Kopparty, Ta-Shma, and Yakirevitch \cite{KTY24,KTY25} as well as an earlier paper by Vl{\u{a}}du{\c{t}} \cite{MR1464542} are the most relevant to our work. We  discuss  Vl{\u{a}}du{\c{t}}'s work in \cref{result:exp sums} and \cite{KTY24,KTY25} in \cref{sec:comparision}.

To clarify the challenges in analyzing trace codes of AG codes and to place our technical results in context, we must first turn to describe the simplest case: the trace of Reed–Solomon codes. This (by now standard) analysis exposes the intimate relationship between taking traces of codes—AG codes or otherwise—and the geometry of algebraic curves. It also shows why existing fundamental results from the theory of algebraic curves, while effective for Reed--Solomon codes, fall short for trace codes of general AG codes. With this perspective, we present our main result, which constitutes a first step toward overcoming these obstacles.

\paragraph{Concurrent work.}
In concurrent, independent work, Kopparty, Ta-Shma, and Yakirevitch \cite{KTY25} also study trace codes of AG codes, focusing on the Hermitian function field; this extends their earlier paper~\cite{KTY24}, which inspired the present work. We give a technical comparison with~\cite{KTY24,KTY25} after presenting our results (see \cref{sec:comparision}).

\subsection{Trace Codes of Reed--Solomon: Analysis via the Hasse--Weil Theorem}\label{sec:intro rs trace}

We recall that a Reed--Solomon (RS) code over $\FF_q$ is defined by identifying messages with polynomials $f\in\FF_q[x]$ of degree $<k$ and evaluating them at distinct field elements. Specifically, for evaluation points $\p_1,\dots,\p_n\in\FF_q$ (with $n\le q$), the codeword corresponding to $f$ is
$
(f(\p_1),\ldots,f(\p_n))
$.
Let $p$ denote the characteristic of $\FF_q$, and write $q=p^m$. The trace code of the Reed--Solomon code is obtained by applying the (absolute) field trace $\tr$ to each coordinate. Thus the encoder maps $\FF_q^k \to \FF_p^n$ via
\begin{equation}\label{eq:tr of f}
f \longmapsto \big(\tr(f(\p_1)),\ldots,\tr(f(\p_n))\big),
\end{equation}
where $\tr(x)=x+x^p+\cdots+x^{p^{m-1}}$. Note that the resulted code is $\FF_p$-linear. As it stands, as long as the encoder is injective, the rate of the code is
\begin{equation}\label{eq:intro wrong rate}
\rho = \frac{k \log q}{n \log p} = \frac{mk}{n},
\end{equation}
which is a factor-$m$ improvement over the rate of the underlying RS code. We now turn to the distance analysis, which is more challenging.\footnote{As it turns out, analyzing the distance will require a slight tweak to the construction, incurring a small loss in the rate computed in \cref{eq:intro wrong rate}.}

Fix a polynomial $f \in \FF_q[x]$ of degree $t<k$. Using Hilbert's Theorem~90, it can be shown that the number of zeros $z_f$ in the codeword corresponding to $f$ in \cref{eq:tr of f} is related to the number of pairs $(x,y)\in\FF_q^2$ satisfying
\begin{equation}\label{eq:intro first as}
y^p - y = f(x),
\end{equation}
that is, to the number of $\FF_q$-rational points $n_f$ on this curve. More precisely, we have that
\begin{ineq}\label{intro:zeros and N}
z_f \leq \frac{n_f}{p},
\end{ineq}
and equality holds when $n=q$.
Therefore, analyzing the distance of the trace code is \emph{equivalent} to counting points on curves, one curve for each message $f$.

Counting points on curves over finite fields is a difficult task. A deep and powerful result, the Hasse--Weil theorem (see \cite[Chapter 5]{Stich}), provides a bound on the number of points. We will discuss the Hasse--Weil theorem in more depth later on; for now it suffices to note that in our current setting, where the curve has the form of \cref{eq:intro first as}, the theorem implies that the number of $\FF_q$-rational points $n_f$ satisfies
\begin{ineq}\label{ineq:intro hw for special case}
|n_f - (q+1)| \le (t-1)(p-1) \sqrt{q}.~\footnote{
To be precise, for some polynomials $f$ the Hasse--Weil theorem does not apply: If $f$ can be written as $f(x) = g(x)^p - g(x)$ for some polynomial $g(x)$, then the number of solutions of \cref{eq:intro first as} is $q \cdot p$, yielding a trivial bound on the number of zeros. To circumvent this, we require the degree $t$ of $f$ to be coprime to $p$. This accounts for the rate loss mentioned above, namely a multiplicative $(1-\frac{1}{p})$ factor relative to \cref{eq:intro wrong rate}.}
\end{ineq}

Combining this with \cref{intro:zeros and N}, we readily obtain a bound on the distance. Assuming, for simplicity, that we evaluate over all field elements (i.e., $n=q$), we get
\[
\delta \ge 1-\frac{1}{p} -\frac{1}{np} - \frac{p-1}{p} \cdot \frac{k-2}{\sqrt{n}} \geq 1 - \frac{1}{p} - \frac{k}{\sqrt{n}}.
\]
Since one cannot expect a distance better than $1-\frac{1}{p}$, we see that the “loss” term is $\frac{k}{\sqrt{n}}$. This, in particular, means that to obtain a nontrivial bound on the distance we must take $k = O(\sqrt{n})$, which forces the rate to vanish at an inverse–square-root rate, $\rho = O\!\left(\frac{1}{\sqrt{n}}\right)$.

\medskip

The above construction  extends naturally to trace codes of AG codes, and this is what we undertake in the following sections. However, the distance analysis based on the Hasse--Weil theorem does not carry over. This motivates our results, and in particular our main theorem: a Hasse--Weil–type bound suitable to the analysis of trace codes of AG codes, which yields a meaningful lower bound on their distance.

To keep this introductory section accessible, we continue to assume no prior knowledge of algebraic function fields. As a result, this section has an expository flavor, while the formal treatment appears later in the technical sections. In particular, in \cref{sec:breif intro curves} we provide a brief, informal introduction to algebraic curves.

\subsection{A Brief Introduction to Algebraic Curves}\label{sec:breif intro curves}

Informally, an algebraic curve over a finite field $\FF_q$ is the set of points in $\FF_q^m$ satisfying $m-1$ independent polynomial equations. A good example to have in mind is the Hermitian plane curve over a field of size $q=r^2$ for some prime power $r$, consisting of all points $(x,y)\in\FF_q^2$ satisfying $y^r + y = x^{r+1}$. It is not hard to show that this curve has $r^3 = q^{3/2}$ points. In general, the number of  points $n$ on a curve is a key parameter; in the coding-theoretic context, it governs the block length of the associated code, as we shall see. Generally, these points are denoted by $\p_1,\ldots,\p_n$.

AG codes are obtained by evaluating functions on a fixed curve. The notion of a function on a curve is somewhat delicate. For example, on the Hermitian curve the two polynomials $y^r$ and $x^{r+1}-y$, although different as formal polynomials, are identical as functions, since they agree on all points of the curve. Nonetheless, there is a notion analogous to the \emph{degree} of a function irrespective of its representation as a polynomial. As with polynomials, this notion of degree bounds the number of zeros a function may have on the curve. Moreover, the set of all functions of degree at most $k$ forms an $\FF_q$-vector space. 

Given a curve, we associate the corresponding AG code in a manner analogous to the Reed--Solomon code. We fix a degree $k$ and identify the messages with the subspace of functions of degree at most $k$. The codeword corresponding to such a function $f$ is obtained by evaluating $f$ at all points on the curve,
$
(f(\p_1),\ldots,f(\p_n)) \in \FF_q^n.
$

A second important parameter is the curve's \emph{genus}. This natural number, denoted by $g$, is a measure of the curve’s complexity. We will not give a formal definition of the genus here, but rather adopt the following operative perspective, which is based on the Hasse-Weil Theorem. The latter is fundamental in the study of algebraic curves over finite fields; it is, in fact, the proof of the Riemann Hypothesis over finite fields. We will discuss the theorem itself later; however, an important corollary---referred to here as the Hasse--Weil bound---states that the number of points $n$ on a curve of genus $g$ satisfies
\[
|n - (q+1)| \le 2\sqrt{q}\cdot g.~\footnote{The reason it is $q+1$ rather than $q$ is that, in this context, curves are taken to be projective rather than affine, so one additional “point at infinity” is included.} 
\]
From this result we deduced \cref{ineq:intro hw for special case}.
Thus, the smaller the genus, the better the bound on the number of points on the curve. That is, we view the genus as the parameter that controls $|n-(q+1)|$ (up to the factor $2\sqrt{q}$).

\subsection{TAG Codes}\label{sec:intro trace of ag}

With the concepts in \cref{sec:breif intro curves} in place, we are ready to give an informal definition of trace codes of AG codes (TAG codes, for short). Fix an algebraic curve over a finite field $\FF_q$ of characteristic $p$, choose $n$ points $\p_1,\ldots,\p_n$, and identify the message space with functions of degree less than $k$. As in \eqref{eq:tr of f}, we map a message $f$ by
\begin{equation}\label{eq:tr of f ag}
f \longmapsto \big(\tr(f(\p_1)),\ldots,\tr(f(\p_n))\big) \in \FF_p^n.
\end{equation}

As in the analysis of the trace of RS codes in \cref{sec:intro rs trace}, here too there is a precise relation between the number of zeros $z_f$ in the codeword corresponding to $f$ and the number of points on a certain curve. This time, the latter curve depends on both the original curve and the function $f$. In particular, analogously to \cref{eq:intro first as}, if the original curve lies in an ambient space of dimension $m$, then the new curve we consider lies in an ambient space of dimension $m+1$, and in addition to the $m-1$ polynomial relations among $x_1,\ldots,x_m$ dictated by the original curve, we have the relation
\begin{equation}\label{eq:m levels}
x_{m+1}^p - x_{m+1} = f(x_1,\ldots,x_m).
\end{equation}

Analogously to \cref{intro:zeros and N}, there is an exact relation between the number of points $n_f$ on the new curve and $n$, the number of points on the original curve. It takes a slightly less simple form, which is quantitatively almost identical; for simplicity, in this informal section we will use the same relation as before, namely,
\begin{equation}\label{intro:zeros and N ag}
z_f = \frac{n_f}{p}.
\end{equation}

Recall that, to analyze the distance of the trace of RS codes, we relied on the Hasse--Weil bound as stated in \cref{ineq:intro hw for special case}. Looking more closely at that equation, the term $q$ counts the number of points on the base curve underlying the Reed--Solomon code, namely the line over $\FF_q$. The Hasse--Weil bound controls the difference between the number of points on the base curve and the number of points on the new curve. The bound in \cref{ineq:intro hw for special case} is in terms of the genus of the new curve.
What should the bound be in the case of TAG codes? There are now two genera involved: the genus of the curve defining the code, denoted by $g$, and the genus of the new, extended curve, denoted by $g_f$. In the original case, since the genus of the line is $0$, it made no appearance in the Hasse--Weil bound.
There is a sense in which the $\sqrt{q}$ appearing in the Hasse--Weil bound provides the natural “units” for measuring the number of points, independent of whether we extend the line or another curve (see \cref{result:exp sums}). Thus, it is conceivable that the difference of genera should be used in the bound. Indeed, the general result in this context is Grothendieck’s trace formula \cite{Grothendieck1977Lefschetz} which states exactly this; namely, it yields the bound
\begin{ineq}\label{ineq:grothendieck}
|n_f - n| \le 2\sqrt{q} \cdot (g_f - g).
\end{ineq}

Unfortunately, this bound, which is tight in general, doesn't give any nontrivial bounds on the distance of TAG codes. To see why, we rely on two results from the theory of algebraic curves. First, there is a formula relating the genus $g_f$ of the extended curve and the genus $g$ of the base curve, known as the Hurwitz genus formula. For our introductory purposes it suffices to say that
\begin{equation}\label{eq:first delta}
g_f = p g + \Delta,
\end{equation}
where $\Delta \ge 0$ is an integer.~\footnote{For readers familiar with algebraic function fields, $\Delta$ is the degree of the Different divisor of the corresponding function field extension.} Hence, the bound given by \cref{ineq:grothendieck} is at least $2\sqrt{q}\,(p-1)g$.

The second result is the Drinfel'd--Vl\u{a}du\c{t} bound \cite{DrinfeldVladut1983FAIA}, which, informally, states that asymptotically, as the genus of the curve tends to infinity, we have
\begin{ineq}\label{ineq:Ihara_bound}
    \frac{n}{g} \le \sqrt{q}-1.
\end{ineq}
Combining this with the calculation above shows that the bound in \cref{ineq:grothendieck} is worse than $2(p-1)n$. Consequently, the resulting upper bound on $n_f$ is no better than $(2p-1)n$. Together with \cref{intro:zeros and N ag}, this yields an upper bound of
$
\frac{(2p-1)n}{p}
$
on the number of zeros in a codeword, which is trivial since this quantity exceeds $n$.\footnote{
We remark that even a bound of the form $\sqrt{q}\,(g_f-g)$ in \cref{ineq:grothendieck}$\,$---which may be obtainable asymptotically---still does not yield a nontrivial distance bound.}

As discussed above, a distance analysis for TAG codes demands estimates stronger than those yielded by Grothendieck’s trace formula. This is the main technical contribution of the present work. In the next section we state our main result and its applications to the analysis of TAG codes, as well as its broader consequences.

\section{Our Results}

As discussed in \cref{sec:intro trace of ag}, the issue with Grothendieck’s trace formula for analyzing the distance of TAG codes is that the upper bound it yields for $|n_f - n|$ depends too heavily on the genus $g$ of the underlying curve. In other words, regardless of how simple the extension is (equivalently, regardless of the degree $t$ of $f$), the bound is too loose when the underlying function field is of large genus. Indeed, the degree of $f$ is encoded in $\Delta$, and the previous section showed that even if the contribution of $\Delta$ to the bound is ignored, one still does not obtain a meaningful bound on the distance. Thus, we seek a bound that depends on the complexity of the \emph{extension} as captured by $\Delta$, rather than on how complex is the base curve, as encoded by the genus $g$. More precisely, one can show that $\Delta = \Theta(pt)$, and in particular it can be taken quite significantly smaller than the genus $g$.

A bound of the form
\begin{equation}\label{eq:Better error term approach}
  |n_f - n| \;=\; \Phi(\Delta)\,\sqrt{q},
\end{equation}
for some function $\Phi$ (e.g., linear in $\Delta$) is most desirable, as it is completely independent of the genus $g$. Our main technical contribution is a step toward such a result, in which the bound we obtain is the \emph{geometric mean} of $\Delta$ (the quantity we wish to appear in the bound) and the genus $g$, which by itself is too large.

For the analysis of TAG codes we only need to consider specific types of extensions, as in \cref{eq:m levels}. Such extensions are called Artin--Schreier extensions. Our result, however, holds in much greater generality. In this introductory section we choose to be somewhat informal and consider only the special case required for analyzing TAG codes.

\begin{theorem}[main result; informal]\label{thm:intro bound informal}
Let $\FF_q$ be a finite field of characteristic $p$. Let $C$ be a ``nice'' curve over $\FF_q$ with genus $g$ and $n$ points. Let $f$ be a ``nice'' function on $C$ of degree $t$, and define the curve $C_f$ by the additional polynomial constraint $y^p - y = f$, where $y$ is a new formal variable. Then the number of points $n_f$ on the curve $C_f$ satisfies
\begin{equation}\label{eq:simplified bound}
    |n_f - n| = O\!\left(p^2\sqrt{t g} \sqrt{q}\right).
\end{equation}
\end{theorem}

In fact, the result also extends to Kummer extensions and, more generally, to extensions of the form \(\phi(y)=f\) for an arbitrary polynomial \(\phi\), as formalized in~\cref{thm:extensions we work with}. The theorem further applies to even more general extensions under technical conditions—hidden in the two “nice” instances highlighted in the informal~\cref{thm:intro bound informal}. For the complete, formal statement, see~\cref{prop:After beta and mu instantiation}.

\subsection{Implications for Error Correcting Codes}
As discussed in \cref{sec:intro trace of ag}, an upper bound on $n_f$ directly translates into a lower bound on the weight of the codeword corresponds to $f$ in the TAG construction. In this section, we examine the resulting codes parameters, also recognizing a possible limitation of this approach. We start by giving an example.

\subsubsection{The Hermitian TAG code}
In this section, we illustrate \cref{thm:intro bound informal} by giving an explicit instantiation of TAG codes based on the Hermitian curve introduced in~\cref{sec:breif intro curves}.
Let $r = p^\ell$ be a power of a prime $p$, and let $q = r^2$. Let $A \subseteq \FF_q \times \FF_q$ be the set of roots of the polynomial $y^r + y - x^{r+1}$. For $T \geq 0$, let
\begin{equation}\label{def:Hermitian B}
\begin{split}
    B = \{ (i, j) \in \mathbb{N} \times \mathbb{N}:\ &ir + j(r+1) \leq T, \\
    &i \text{ is odd and } j \text{ is even}, \\
    &j < r/6\}.
\end{split}
\end{equation}
The Hermitian TAG code is defined by
\[
    \mathcal{C} = \left\{ \left(\tr \left( \sum_{(i,j)\in B} c_{i,j}\alpha^i \beta^j \right) \right)_{(\alpha, \beta) \in A} : \ c_{i,j} \in \FF_q \right\}.
\]
The following theorem specifies the parameters of this code; the formal statement appears in~\cref{thm:epsilon-balanced codes from Hermitian curve}.

\begin{theorem}[Hermitian TAG codes; informal]The Hermitian TAG code with message length $k$ and relative distance $1/2-\eps$ has block length
    \[
        n = O\left ( \frac{k}{\varepsilon^4} \right )^{3/2}.
    \]
\end{theorem}
In fact, the result shows that the Hermitian TAG code not only has relative distance \( \tfrac12-\varepsilon \) but is actually \emph{\(\varepsilon\)-balanced}: every nonzero codeword has relative Hamming weight in \([\,\tfrac12-\varepsilon,\tfrac12+\varepsilon\,]\). In \cref{sec:specific function fields}, we further instantiate our results for additional curves—specifically the Hermitian tower of function fields and the norm–trace function field—obtaining \(\varepsilon\)-balanced codes with varying parameter trade-offs.

\subsubsection{The high distance regime: TAG codes vs.\ concatenation}\label{sec:barrier intro}
As our case studies in~\cref{sec:specific function fields} show, in the regime \(\delta=\tfrac{1}{2}-\varepsilon\) all the TAG codes we consider are still far from the GV bound, requiring
$
n=\omega\!\left(\tfrac{k}{\varepsilon^{2}}\right).
$
This prompts the question: does there exist a curve for which the corresponding TAG code matches the GV bound?

In~\cref{sec:trace_vs_had} we present strong evidence for a negative answer.
 We show that, when analyzed using our bound from~\cref{thm:intro bound informal}, in the high-distance regime \(\delta=\tfrac{1}{2}-\varepsilon\), TAG codes instantiated from a given AG code are outperformed by concatenating the same AG code with Hadamard. This remains true even under a bound of the form~\cref{eq:Better error term approach} with \(\Phi(\Delta)\) that is linear in \(\Delta\). In particular, we establish the following.

\begin{theorem}\label{prop:had vs tag}
    Assume the bound given by \cref{eq:Better error term approach} is tight up to a constant. Then, any TAG code with message length $k$ and relative distance $1/2 - \varepsilon$ has block length 
    \[
        n = \Omega\left( \frac{k}{\varepsilon^{3}} \right).
    \]
    Moreover, assuming the bound given in \cref{thm:intro bound informal} is tight up to a constant,
        $n = \Omega(k/\eps^6)$.
\end{theorem}

\subsubsection{The constant distance regime}\label{subsubsec:constant rate TAGs}

As implied by \cref{sec:barrier intro}, under our analysis, attaining the GV bound requires operating in the regime where the distance \(\delta\) is bounded away from \(\tfrac12\).
To approach the GV bound in this regime, we must have a nonvanishing rate. In this short section, we examine the implications for the parameters of the underlying AG code.

For a curve \(C\), let \(\ell(T)\) denote the dimension of the vector space of functions on \(C\) of degree less than \(T\). Assume that \cref{thm:intro bound informal} applies to every function in this space, and let \(\mathcal{C}\) be the resulting TAG code. By \cref{eq:simplified bound} together with \cref{intro:zeros and N ag}, the encoding of a function \(f\) of degree \(t\) is nonzero provided the right-hand side is \(<(p-1)\,n\).

For simplicity, assume the implicit constant hidden in the big-\(O\) of \cref{eq:simplified bound} is \(1\), and that \(n/g=\sqrt{q}-1\) (the optimal AG-code guarantee, meeting the Drinfel'd-Vl\u{a}du\c{t} bound). To obtain a nontrivial bound on the number of zeros of \(f\), we require
\[
  p^{2}\sqrt{tg}\,\sqrt{q}<pn,
\]
which is equivalent (since \(n=g(\sqrt{q}-1)\)) to
\[
  t < \frac{g}{p^{2}}\left(1-\frac{1}{\sqrt{q}}\right)^{2}.
\]
Hence, to apply our result we must work in the regime \(T < g/p^{2}\). The resulting TAG code \(\mathcal{C}\) has rate
\[
  \rho \;=\; \frac{\ell(T)\,\log_{p}q}{n}.
\]
Since \(\ell(T)\le T\le g/p^{2}\) and \(n=(\sqrt{q}-1)g\), we obtain
\[
  \rho \;=\; \frac{\log_{p}q}{\sqrt{q}-1}\cdot\frac{\ell(T)}{g}
  \;\le\; \frac{\log_{p}q}{(\sqrt{q}-1)\,p^{2}}.
\]
Thus, to obtain a family with constant rate, we must fix \(q\) to a constant and choose curves (and \(T\le g/p^{2}\)) so that \(\ell(T)=\Omega_{q}(g)\).

This is a somewhat nonstandard regime for AG codes. Typically one works with functions of degree at least \(2g\), where Riemann--Roch guarantees linear growth of \(\ell(T)\) (indeed, \(\ell(T)=T+1-g\) for \(T\ge 2g-1\)). In particular, among degrees up to \((1+\alpha)g\) one obtains at least \(\alpha g\) attainable degrees. By contrast, working below \(g\) is subtler: the attainable sub-\(g\) degrees—the \emph{Weierstrass non-gaps}—depend on the curve. In our TAG setting, one must therefore understand this sub-\(g\) regime.

A property closely related to the latter was examined in~\cite[Section~4.1]{Ben-Aroya_Ta-Shma}.
However, we are not aware of any optimal family of function fields exhibiting this behavior. For instance, in~\cite{yang2019weierstrasssemigroupstowerfunction}, the authors show that certain Riemann--Roch spaces on the Garcia--Stichtenoth tower do not satisfy this property. We formalize this  as an open problem.

\begin{openproblem}
    Does there exist a prime power $q=p^\ell$ and a family of function fields $F_i / \FF_q$ with $\frac{n_i}{g_i} = \Theta_q(1)$, such that each $F_i$ of genus $g_i$ contains a Riemann--Roch space $\mathcal{L}_i$ of degree $T_i \leq \frac{g_i}{p^2}$ satisfying $\ell_i(T_i) = \Omega_q(g_i)$?
\end{openproblem}

\subsubsection{Limitations and consequences of improving the bound}

In \cref{subsubsec:constant rate TAGs}, we assumed that the constant appearing in the error term in \cref{eq:simplified bound} is equal to~1.
Suppose instead that we could obtain a small constant~$0<c<1$ such that, for every ``nice'' function~$f$ of degree~$t$, the bound
\[
    n_f \leq n + c p^2 \sqrt{t g} \sqrt{q}
\]
holds. If a $\beta$-fraction of all functions of degree up to~$T$ would be ``nice'', then the resulting TAG code would have rate and relative distance, correspondingly, 
\begin{align*}
    \rho &= \beta \cdot \frac{\log_p q}{\sqrt{q} - 1} \cdot \frac{\ell(T)}{g}, \\
    \delta &\geq 1 - \frac{1}{p} - c \cdot \frac{p \sqrt{gT}\sqrt{q}}{n}.
\end{align*}
Taking $T = 2g$, the Riemann--Roch theorem gives $\ell(T) = g + 1$. 
We then have
\begin{align*}
    \rho &\geq \beta \cdot \frac{\log_p q}{\sqrt{q} - 1}, \\
    \delta &\geq 1 - \frac{1}{p} - c p \frac{\sqrt{2q}}{\sqrt{q} - 1}.
\end{align*}

From this we see that if our bound were sharpened by reducing the leading constant, TAG codes would work in the usual AG regime (e.g., for degrees \(\ge 2g-1\)) and could lead to better code parameters.
Note that the constant \(c\) cannot be taken arbitrarily small; doing so would violate the MRRW linear-programming upper bounds~\cite{MRRW}. This offers an unusual direction of inference: using coding-theoretic limits to deduce statements about algebraic curves.

\subsection{Exponential Sums over Curves}\label{result:exp sums}
Exponential sums over finite fields play a central role in number theory and algebraic geometry.  
A.~Weil was the first to observe a deep connection between exponential sums of polynomials and Artin--Schreier covers.  
Specifically, for a polynomial $f(x)$ of degree $d$ over $\FF_q$ such that $f \neq u^p - u$ for all $u \in \overline{\FF_q}(x)$, the exponential sum
\begin{equation}
    S(f) = \sum_{\alpha \in \FF_q} e^{\frac{2\pi i}{p}\tr(f(\alpha))}
\end{equation}
can be expressed as
$S(f) = \omega_1 + \cdots + \omega_D$, 
where the $\omega_i$ are a subset of the reciprocals of the roots of the zeta function of the Artin--Schreier curve
\[
    y^p - y = f(x).
\]

Weil also showed that if $\deg f$ is relatively prime to~$p$, then $D = d - 1$.  
Combined with his proof of the Riemann Hypothesis for curves over finite fields, which asserts that each $\omega_i$ satisfies $|\omega_i| = \sqrt{q}$, this yields the bound
\[
    |S(f)| \leq (d - 1)\sqrt{q}.
\]
In 1966, Bombieri~\cite[Theorem~5]{Bom1} extended Weil’s results from polynomials to general rational functions defined over algebraic curves.

\begin{theorem}[\cite{Bom1}, Theorem~5]\label{thm:bom-exp-sum-curves}
    Let $C$ be a complete, irreducible, non-singular curve of genus $g$ over $\FF_q$, which contains $n$ points. 
    Let $f \in \FF_q(C)$ be a function on $C$ such that $f \neq u^p - u$ for all $u \in \overline{\FF_q}(C)$. 
    Let $\mathcal{P}$ be the set of poles of $f$. 
    Then the exponential sum
    \begin{equation}\label{eq:def_exp_sum_curve}
        S(f) = \sum_{\p \in C \setminus \mathcal{P}} e^{\frac{2\pi i}{p}\tr(f(\p))},
    \end{equation}
    can be expressed as
    \begin{equation}
        S(f) = \omega_1 + \cdots + \omega_D,
    \end{equation}
    where $|\omega_i| = \sqrt{q}$ for all $1 \leq i \leq D$, and 
    \[
        D \leq 2g - 2 + |\mathcal{P}| + \deg(f).
    \]
    Moreover, if $f$ has a single pole of degree relatively prime to $p$, then 
    $D = 2g - 1 + \deg(f)$.
\end{theorem}

\begin{corollary}[\cite{Bom1}]\label[corollary]{cor:bom_bound_curve}
    In the setting of \cref{thm:bom-exp-sum-curves}, we have
    \[
        |S(f)| \leq (2g - 1 + \deg(f)) \sqrt{q}.
    \]
\end{corollary}

Note that a trivial bound $|S(f)| \leq n - |\mathcal{P}|$ follows directly from \cref{eq:def_exp_sum_curve}. 
Hence, the bound in \cref{cor:bom_bound_curve} is non-trivial only when
\begin{ineq}\label{ineq:cond_non_trivial_bom_result}
    (2g - 1 + \deg(f)) \sqrt{q} < n - |\mathcal{P}|.
\end{ineq}
\cref{cor:bom_bound_curve} was used in~\cite{trace_code_dimension} to compute the dimension of certain TAG codes defined over curves for which \cref{ineq:cond_non_trivial_bom_result} holds.

There is a major limitation of Bombieri's bound in our setting:
    If the curve $C$ satisfies $2g\sqrt{q} \geq n$, then the bound becomes trivial for any function on the curve. 
    By the Drinfeld--Vl\u{a}du\c{t} bound, \cref{ineq:Ihara_bound}, this inequality holds for all curves of sufficiently large genus. 
    Since such curves are precisely those used in AG codes constructions, Bombieri's bound is not applicable in this regime.

In \cite{MR1464542}, Vl{\u{a}}du{\c{t}} proves nontrivial bounds for the case of Hermitian and Hansen–Stichtenoth curves, showing that for certain functions \(f\) of bounded degree one has \(S(f) \le cn\) for some constant \(0 < c < 1\).  Our results provide new bounds on exponential sums of functions over curves of large genus, yielding in particular that \(S(f) = o(n)\) for functions \(f\) of degree \(o(g)\).

\begin{theorem}[exponential sums; informal]\label{thm:exp sums bound informal}
Let $\FF_q$ be a finite field of characteristic~$p$. 
Let $C$ be a ``nice'' curve over $\FF_q$ with genus~$g$ and~$n$ rational points, and let $f$ be a ``nice'' function on~$C$ of degree~$t$. 
Then,
\begin{equation*}
    S(f) = O\!\left(p^3 \sqrt{t g} \sqrt{q}\right).
\end{equation*}
\end{theorem}
The reader is referred to \cref{thm:exponential_sums} for the formal statement.
Interestingly, in conjunction with \cref{thm:bom-exp-sum-curves}, this implies that the complex roots~$\omega_i$, each of absolute value~$\sqrt{q}$, cannot be all in the same direction.

Note that this bound remains meaningful even for families of curves with $g \to \infty$. 
However, since $g\sqrt{q} = \Theta(n)$, as in \cref{subsubsec:constant rate TAGs}, the function~$f$ must satisfy $t \ll g$ for the bound to be non-trivial.

\subsection{Comparison with \cite{KTY24,KTY25}}\label{sec:comparision}

In this section, we compare our results with those obtained in an earlier work by Kopparty, Ta-Shma, and Yakirevitch~\cite{KTY24}, and with a concurrent work by the same authors~\cite{KTY25}. These works focus primarily on the Hermitian function field. For a technical comparison, we instantiate~\cref{prop:Kummer bound before instantiation} with this function field.
 
\begin{theorem}\label{prop:kummer Hermitian}
    Let $d > 1$ be such that $d\mid q-1$. Let $f \in \rrs(q \mathfrak{P}_\infty)$ of degree $t = ir + j(r+1)$, such that $(d, t) = 1$, $(i+j, d) = 1$ and $j < \frac{r}{3d}$. Then,
    \[
        \frac{1}{N} \left| \left\{ \q \in \mathbb{P}^1_F \setminus \{ \p \} : \exists\, y \in \FF_q^\times,\ y^d = f(\q ) \right\}\right| \leq \frac{1}{d} + O\left ( p \sqrt{\frac{t}{q}} + p d\, \frac{t}{q} \right ).
    \]
\end{theorem}

The corresponding theorem from \cite{KTY24} goes as follows.

\begin{theorem}[\cite{KTY24}]\label{thm:Kedem Hermitian}
    Assume that $r > 500$ is a prime number, $q = r^2$ and let $d$ be a prime number that divides $q - 1$.  Let $f \in \rrs(q \mathfrak{P}_\infty)$ of degree $t = ir + j(r+1)$, such that $(d, t) = 1$, $(i+j, d) = 1$, and $d < O(\sqrt{t} + \frac{q^{3/2}}{t})$. Then,
    \[
        \frac{1}{N} \left| \left\{ \q \in \mathbb{P}^1_F \setminus \{ \p \} : \exists\, y \in \FF_q^\times,\ y^d = f(\q ) \right\}\right| \leq \frac{1}{d} + O\left ( \sqrt{\frac{t}{q}} \right ).
    \]
\end{theorem}
There are several distinctions between these results. First, \cref{thm:Kedem Hermitian} is restricted to the case where $r$ is a prime number, reflecting the limitations of the derivatives method as used in \cite{KTY24}. In contrast, our result applies to general prime powers $r = p^e$; but the error term in our bound depends on $p$.
For the purpose of TAG codes, one can take $p$ to be a constant, particularly $p=2$ (and we are free to vary $r$) so it has no effect on our bound. Importantly, our result is meaningful for all $e > 2$ whereas \cref{thm:Kedem Hermitian} requires $e = 1$.

In~\cite{KTY25}, the authors extend the bound for the case \(d = 2\) to all functions \(f\in \mathcal{L}\left(q^{3/4}\,\mathfrak{P}_\infty\right)\) for which the polynomial \(T^{2} - f\) is absolutely irreducible, without imposing any restriction on \(\deg(f)\).  
However, this comes at the cost of replacing the error term \(O\!\left(\sqrt{\frac{t}{q}}\right)\) with \(O\!\left(\frac{t}{q^{3/4}}\right)\). Moreover, in \cite{KTY25}, the authors provide a bound on the exponential sum \cref{eq:def_exp_sum_curve} in the setting of the Hermitian curve, when \(p = 2\) and \(q\) is an arbitrary power of \(2\), for any function \(f\in \mathcal{L}\left(q^{3/4}\,\mathfrak{P}_\infty\right)\) of odd degree; the resulting error term is \(O\!\left(\tfrac{t}{q^{3/4}}\right)\), whereas our result \cref{thm:exponential_sums} yields the sharper bound \(O\!\left(\sqrt{\tfrac{t}{q}}\right)\) as $t = \Omega(\sqrt{q})$.

\section{Preliminaries}\label{sec:preliminaries}
We assume familiarity with basic background on algebraic function fields, such as places, valuations, extensions, decomposition and ramification of places, etc. A detailed exposition of the subject can be found in \cite{Stich}. In this section we recall some basic notions concerning algebraic function fields, the trace map over finite fields, $\varepsilon$-balanced codes, and code concatenation.

\subsubsection*{Algebraic function fields, valuations and places}
\label{subsec:function-fields}

Let $\FF_q$ be the finite field with $q$ elements. The \emph{rational function field} $\FF_q(x)$
is the field of rational functions in an indeterminate $x$ with coefficients in $\FF_q$. An \emph{algebraic function field} $F/\FF_q$ is a finite algebraic extension of $\FF_q(x)$.  Elements of $F$ are called \emph{functions}.

A \emph{discrete valuation} on $F$ is a map $v \colon F^\times\to \ZZ$ satisfying
$v(fg)=v(f)+v(g)$, $v(f+g)\ge\min\{v(f),v(g)\}$, and it can be extended to $F$ by setting $v(0) = \infty$ with the understanding that $\infty>n$ for all $n\in\ZZ$.
Associated to a discrete valuation $v$ is its valuation ring
$\mathcal{O}_v=\{f\in F: v(f)\ge 0\}$, a maximal ideal $\mathfrak m_v=\{f\in\mathcal{O}_v: v(f)>0\}$ and the
residue field $\FF_v=\mathcal{O}_v/\mathfrak m_v$.  A \emph{place} $\p$ of $F$ is the maximal ideal $\mathfrak m_v$ of some valuation ring $\mathcal{O}_v$.  We write $v_\p$ for the valuation corresponding to $\p$,
$\mathcal{O}_\p$ for its valuation ring and $\FF_\p$ for its residue field.

The \emph{degree} of a place $\p$ is $\deg \p := [\FF_\p:\FF_q]$. Places of degree $1$ are called
\emph{rational places}.  We denote the set of all places of $F$ by $\mathbb{P}_F$, and the set of all rational places of $F$ by $\mathbb{P}^1_F$.

\subsubsection*{Extensions of function fields and integral closures}
\label{subsec:function-field-extensions}

Let \( F/\FF_q \) be an algebraic function field, and let \( L/F \) be a finite extension of function fields. 
For each place \( \p \in \mathbb{P}_F \), we consider the behavior of \(\p\) in the extension \( L/F \).
A place \( \mathfrak{P} \in \mathbb{P}_L \) is said to \emph{lie above} \( \p \), denoted  $\pp \,|\, \p$, if 
$\mathcal{O}_\p \subseteq \mathcal{O}_{\mathfrak{P}}$,
or equivalently, if \( \mathfrak{P} \cap \mathcal{O}_\p = \p \).  
For a fixed place \( \p \) of \( F \), there are finitely many places \( \mathfrak{P}_1, \ldots, \mathfrak{P}_r \) of \( L \) lying above it.
Associated with each \(\pp\,|\,\p\) are two important parameters:
\begin{itemize}
    \item The \emph{ramification index} \( e(\mathfrak{P}|\p) \), defined to be the integer satisfying \( v_{\mathfrak{P}}(h) = e(\mathfrak{P}|\p)\, v_\p(h) \) for all \( h \in F \).
    \item The \emph{inertia degree} \( f(\mathfrak{P}|\p) := [\FF_{\mathfrak{P}} : \FF_\p] \).
\end{itemize}
These parameters satisfy the fundamental identity
\[
    \sum_{\mathfrak{P} \mid \p} e(\mathfrak{P}|\p) f(\mathfrak{P}|\p) = [L : F].
\]
The \emph{integral closure} of \( \mathcal{O}_\p \) in \( L \) is defined as
\[
    \mathcal{O}_\p' := \{\, f \in L : f \text{ is integral over } \mathcal{O}_\p \,\}.
\]
Equivalently, \( \mathcal{O}_\p' \) consists of all elements of \( L \) that satisfy a monic polynomial with coefficients in \( \mathcal{O}_\p \). 
\subsubsection*{Divisors}
\label{subsec:divisors}

A \emph{divisor} of $F$ is a formal finite $\ZZ$-linear combination
\[
G = \sum_{\p\in\mathbb P_F} n_\p\,\p,\qquad n_\p\in\ZZ,
\]
with finite support $\support(G)=\{\p\in\mathbb{P}_F:n_\p\neq0\}$. 
For two divisors $G_1,G_2$ we write $G_1\ge G_2$ if $n_\p(G_1)\ge n_\p(G_2)$
for all $\p$.
The \emph{degree} of $G$ is
$\deg(G)=\sum_\p n_\p\deg \p$. 

Every nonzero function $f\in F^\times$ determines the \emph{principal divisor}
\[
(f) := \sum_{\p\in\mathbb{P}_F} v_\p(f)\,\p.
\]
The \emph{pole divisor} (or divisor of poles) of $f$ is
\[
(f)_\infty := \sum_{\p\,:\, v_\p(f)<0} -v_\p(f)\,\p.
\]
The \emph{zero divisor} of $f$ is defined 
analogously by
\[
(f)_0 := \sum_{\p\,:\,v_{\p}(f)>0}v_{\p}(f)\p.
\]
An important connection between those divisors states that the number of zeros equals the number of poles.
Moreover, if $f\in F\setminus\FF_q$ then
\[
    \deg (f)_\infty = \deg (f)_0 = [F : \FF_q(f)].
\]
In particular, $(f)=(f)_0-(f)_\infty$ has degree 0. 
The number 
$$
\gamma = \min\{[F:\FF_q(f)] : f\in F\setminus\FF_q\}
=
\min\{\deg(f)_\infty : f\in F\setminus\FF_q\}
$$
is called the \emph{gonality} of 
$F\slash \FF_q$. 
The following claim gives a lower bound on $\gamma$.
\begin{claim}[\cite{Ben-Aroya_Ta-Shma}, Lemma 4.2]\label[claim]{clm:LB_for_the_gonality}
Let $F\slash\FF_q$ be a function field with $N_F$ rational places. Then
$$ \gamma \geq \frac{N_F}{q+1}. $$
\end{claim}

\subsubsection*{Riemann--Roch spaces and the genus}
\label{subsec:riemann-roch}

For a divisor $G$, define the \emph{Riemann--Roch space}
\[
\rrs(G) \;=\; \{\, f\in F^\times : (f) + G \ge 0 \,\}\cup\{0\}.
\]
This is a finite-dimensional $\FF_q$-vector space; denote $\ell(G):=\dim_{\FF_q}\rrs(G)$.  The function
$\ell(\cdot)$ encodes the number of independent functions with poles bounded by a given divisor.
\begin{claim}\label[claim]{clm:UB_for_rrs_dim}
    Let $G$ be a divisor of $F$ with $\deg(G)\geq 0$. Then $\ell(G)\leq \deg(G)+1$. 
\end{claim}
The \emph{genus} $g=g(F)$ of the function field $F$ is the unique nonnegative integer for which $\ell(G)\ge \deg(G)-g+1$ holds for every divisor $G$, and equality holds
for all $G$ of sufficiently large degree. A divisor $W$ is called \emph{canonical} if $\deg(W)=2g-2$ and $\ell(W)=g$. The precise relation between $\deg(G)$ and $\ell(G)$ is given by the celebrated Riemann-Roch Theorem.

\begin{theorem}[Riemann--Roch]
If $G$ is a divisor of $F$ and $W$ is a canonical divisor, then
\[
\ell(G) - \ell(W-G) \;=\; \deg(G) - g + 1.
\]
In particular, if $\deg(G)\ge 2g-1$ then
$\ell(G)=\deg(G)-g+1$.
\end{theorem}

\subsubsection*{Bounds on the number of rational places}
\label{subsec:hasse-weil}

Let $N_F$ denote the number of rational places of a function field $F/\FF_q$ of genus $g$. The classical Hasse--Weil bound states that
\begin{ineq}\label{ineq:Hasse-Weil}
    \left| N_F - (q+1) \right| \le 2g\sqrt{q}.
\end{ineq}
Thus $N_F$ grows at most linearly with $g$.

Considering the asymptotic regime, let
\[
    N_q(g) := \max \{ N_F\ : \ F \textnormal{ is a function field over } \FF_q \textnormal{ of genus }g\}.
\]
Ihara's constant is defined by
\[
A(q) := \limsup_{g\to\infty} \frac{N_q(g)}{g}.
\]
Drinfeld and Vl\u{a}du\c{t} proved the upper bound
\[
A(q) \le \sqrt{q}-1.
\]
Moreover, when $q$ is a square, this bound is tight: explicit towers of function fields (notably those of Garcia--Stichtenoth) achieve $N_F/g_F\to \sqrt{q}-1$. Such optimal towers are the key source of asymptotically optimal algebraic-geometric codes.

\subsubsection*{Trace map on finite fields}
\label{subsec:trace}

For a prime power $q$ and an extension of finite fields $\FF_{q^m}/\FF_q$, the \emph{trace} 
$\tr_{\FF_{q^m}/\FF_q}\colon \FF_{q^m}\to\FF_q$
is the $\FF_q$-linear map given by
\[
\tr_{\FF_{q^m}/\FF_q}(z) \;=\; z + z^{q} + z^{q^2} + \cdots + z^{q^{m-1}},\qquad z\in\FF_{q^m}.
\]
The trace is surjective and every $a\in\FF_q$ has inverse image $\tr^{-1}(a)$ of size $q^{m-1}$.
\subsubsection*{Additional coding theory preliminaries}
\label{subsec:balanced}

We say that $C$ is an $[n,k,d]_q$ code if $C$ is a linear subspace of $\FF_q^n$ of dimension $k$, and the distance of $C$ (i.e., the minimal Hamming distance between each two distinct codewords) is at least $d$. We will often choose to omit the distance parameter.
We say that an $[n,k]_2$ code $C$ is \emph{$\eps$-balanced} if for any nonzero $c \in C$, the relative Hamming weight of $c$ lies in the range $[1/2-\eps,1/2+\eps]$.\footnote{Linear $\eps$-balanced codes are essentially equivalent to \emph{$\eps$-biased sets}, a primitive of great importance in pseudorandomness (see, e.g., \cite[Section 2.2]{HH24}).}

We conclude this section by recalling the operation of code concatenation.
Letting $C_{\mathrm{out}}\subseteq\FF_q^{N}$ be an ``outer'' $[N,K,D]_q$ code, and $C_{\mathrm{in}}$ be an ``inner'' $[n_{\mathrm{in}},\log q,d_{\mathrm{in}}]_2$ code, the concatenated code $C_{\mathrm{out}}\circ C_{\mathrm{in}}$ is an $[N \cdot n,K \cdot \log_{2}q,D \cdot d]$ such that for any $x \in \FF_{q}^{K} \equiv \FF_{2}^{K \cdot \log_{2}q}$, 
\[
\left(C_{\mathrm{out}}\circ C_{\mathrm{in}}\right)_{i,j} = C_{\mathrm{in}}(C_{\mathsf{out}}(x)_{i})_{j}
\]
where $i \in [N]$ and $j \in [n]$.
A particularly useful inner code is the \emph{Hadamard code}. Given a message length $m$, set $q = 2^m$, and identify $\FF_q$ with the vector space $\FF_2^m$. For $u\in\FF_2^m$, the (binary) Hadamard codeword $\text{Had}(u)$ is the length-$2^m$ vector indexed by $v\in\FF_2^m$ with entries $\langle u,v\rangle$ (dot product modulo $2$). The Hadamard code has relative distance $\tfrac12$.

\section{TAG Codes and Function Field Extensions}\label{sec:trace codes}
In this short section, we formally define the construction of the trace code associated with an algebraic geometric code. We then review the relation of its minimum distance to the number of rational points in a certain elementary abelian $p$-extension of the underlying function field. This connection is a known result, which we include here for completeness. 

Let $p$ be a prime number, and let $p\leq \ell \leq q$ be powers of $p$ such that $\FF_p\subseteq\FF_{\ell}\subseteq\FF_q$. Throughout, $\tr$ is the trace function from $\FF_q$ down to $\FF_\ell$. Let $F/\FF_q$ be a function field of genus $g$ with $N = n+1$ rational points, one of which is denoted by $\p$, and the remaining by $\p_1, \ldots, \p_n$. For an integer $r$, we denote the set of pole numbers up to $r$ at $\p$ by
\[
\wrs_r = \{ i \in [r] : \rrs(i\p) \neq \rrs((i-1)\p) \},
\]
and for each $i \in \wrs_r$, let $b_i \in \rrs(i\p) \setminus \rrs((i-1)\p)$. We now focus on those functions whose pole order is coprime to the characteristic $p$: let $i_1 < i_2 < \cdots < i_k$ be the elements of $\wrs_r$ such that $p \nmid i$, and define
\[
B_r = \{ b_i : i \in \wrs_r \text{ such that } p \nmid i \}
\]
to be the corresponding set of functions.

With these notations, we define the $\FF_\ell$-linear code
$
\trcode \colon \FF_q^k \to \FF_\ell^n
$
as follows. Given $m = (m_1, \ldots, m_k) \in \FF_q^k$, define the function
$
f_m = \sum_{j=1}^{k} m_j b_{i_j}.
$
Then, set
\[
\trcode(m) = \left( \tr(f_m(\p_1)), \ldots, \tr(f_m(\p_n)) \right).
\]
Put differently, we identify the domain of $\trcode$ with $\spn_{\FF_q}(B_r)$, and for a given input function $f \in \spn_{\FF_q}(B_r)$, we evaluate it at the $n$ rational points $\p_1, \ldots, \p_n$, followed by applying the trace function from $\FF_q$ to $\FF_\ell$ coordinate-wise.

We now turn to analyze the distance of the code $\trcode$. As hinted above, the key fact used is that the distance is closely related to the number of rational points in a certain extension of the function field $F$. To make this connection precise, fix a function $0 \neq f \in \spn_{\FF_q}(B_r)$ and consider the field extension $L = F(z)$ defined by the equation 
\[
    z^\ell - z = f.
\]
We first note that $L/F$ is an elementary abelian $p$-extension, which is a generalization of Artin-Schreier extensions (which are the case $\ell = p$). Indeed, notice that the polynomial $a(T) = T^\ell - T$ is additive, and its set of roots is $\FF_\ell \subseteq \FF_q$. In addition, $v_\p(f)<0$ with $p\nmid v_\p(f)$, and $v_\mathfrak{q}(f) \geq 0$ for all $\q\in \mathbb{P}_F\setminus\{\p\}$. Thus, all the conditions of~\cite[Proposition~3.7.10]{Stich} are satisfied, and the extension $L / F$ is an elementary abelian $p$-extension of degree $\ell$.

\begin{claim}\label[claim]{c:splits whenever trace vanishes}
	For every $i \in [n]$, we have $\tr(f(\p_i)) = 0$ if and only if the place $\p_i$ splits completely in the extension $L/F$.
\end{claim}

For the proof of \cref{c:splits whenever trace vanishes}, we make use of Kummer's Theorem (see Theorem~3.3.7 in~\cite{Stich}), which we cite here, with some modifications to suit our needs, for convenience.
\begin{theorem}[Kummer's Theorem]\label{thm:kummers theorem}
	Let $L/F$ be a function field extension, and fix a place $\p$ of $F$. Assume there exists an element $z \in L$ such that $L = F(z)$ and $z \in \vrc_\p$. Consider the minimal polynomial of $z$ over $F$,
	\[
	\varphi(T) = \sum_{i=0}^{d} h_i T^i,
	\]
	where we use the known fact that $h_i \in \mathcal{O}_\p$ for all $i$. Define
	\[
	\bar{\varphi}(T) = \sum_{i=0}^{d} h_i(\p) T^i \in F_\p[T],
	\]
	where $F_\p$ denotes the residue class field at $\p$. Factor $\bar{\varphi}$ over $F_\p$ as
	\[
	\bar{\varphi}(T) = \prod_{j=1}^{r} \gamma_j(T)^{\eps_j},
	\]
	where $\gamma_1, \ldots, \gamma_r$ are distinct irreducible factors and $\eps_j$ denotes the multiplicity of $\gamma_j$ in the factorization. Assuming that $\eps_1 = \cdots = \eps_r = 1$, there are exactly $r$ distinct places $\pp_1, \ldots, \pp_r$ of $L$ lying above $\p$. Moreover, for each $i \in [r]$, we have $e(\pp_i|\p) = 1$ and $f(\pp_i|\p) = \deg \gamma_i$.
\end{theorem}

With this we are in position to prove \cref{c:splits whenever trace vanishes}.

\begin{proof}[Proof of \cref{c:splits whenever trace vanishes}]
	We recall Hilbert's Theorem 90, in its additive form, which asserts that for every $\alpha \in \FF_q$
	$$
	\tr(\alpha) = 0 \quad \iff  \quad \exists \beta \in \FF_q \,\,\,\, \alpha = \beta^\ell - \beta.
	$$
	Fix $i \in [n]$ and assume that $\tr(f(\p_i)) = 0$. By Hilbert's Theorem 90, 
	\begin{equation}\label{eq:f of p is diff of betas}
	\exists \beta \in \FF_q \quad 
	f(\p_i) = \beta^\ell - \beta.
	\end{equation}
	Observe that $z \in \vrc_{\p_i}$. Indeed, the minimal polynomial of $z$ over $F$ is  
    $$
    \varphi(T) = T^\ell-T-f\in \vr_{\p_i}[T],
    $$
    where $f\in\vr_{\p_i}$ since the only pole of $f$ is $\p$.  
    As we also have that $L = F(z)$, we are in a position to apply Kummer's Theorem. With the notations of \cref{thm:kummers theorem}, we get by \cref{eq:f of p is diff of betas} that
	\begin{align*}
	\bar{\varphi}(T) &=
	T^\ell-T-f(\p_i) \\ 
	&= T^\ell-T-(\beta^\ell - \beta) \\
	&= (T-\beta)^\ell - (T-\beta) \\
	&= \prod_{t \in \FF_\ell}(T-\beta-t).
	\end{align*}
	Hence, Kummer's Theorem implies that $\p_i$ splits completely in the extension $L/F$.

    On the other hand, if \( \tr(f(\p_i)) \neq 0 \), then Hilbert's Theorem~90 implies that the polynomial \( \bar{\varphi}(T) \) has no roots in \( \FF_q \). 
    Notice that $F_{\p_i}=\FF_q$, as $\p_i$ is a rational place. Let
    \[
	\bar{\varphi}(T) = \prod_{j=1}^{r} \gamma_j(T)^{\eps_j}
	\]
    be the factorization of $\bar{\varphi}(T)$ over $F_{\p_i}$. Then it follows that $\deg\gamma_j>1$ for all $j\in [r]$. 
    In addition, since $\gcd\left (\bar{\varphi}(T), \bar{\varphi}'(T) \right ) = 1$, the polynomial $\bar{\varphi}(T)$ is separable. Hence $\varepsilon_j= 1$ for all $j\in [r]$. Thus, by Kummer's Theorem, every place of $L$ lying above $\p_i$ must be of degree $ > 1$, and the proof is completed.
\end{proof}

As a consequence of \cref{c:splits whenever trace vanishes}, we obtain the following result. We denote by $N_L$ the number of rational points of $L$.

\begin{corollary}\label[corollary]{cor:distance of trace and as}
	The relative-distance of the code $\trcode$ defined above is 
\[
\delta = 1 - \frac{1}{\ell}\cdot\frac{N_L-1}{n}.
\]
\end{corollary}

\begin{proof}
Choose a codeword  
$\left( \tr(f(\p_1)), \ldots, \tr(f(\p_n)) \right)$ 
of
minimal weight, and let
\[
Z = \{ i \in [n] : \tr(f(\p_i)) = 0 \}.
\]
Then the relative distance of the code is given by $\delta = 1-\frac{|Z|}{n}$. 

By \cref{c:splits whenever trace vanishes}, for every $i \in Z$ there are exactly $\ell$ rational places of $L$ lying above $\p_i$. Moreover, from the proof of the claim, for every $i\in [n]\setminus Z$, there are no rational places of $L$ lying above $\p_i$. As the remaining rational place $\p\in\mathbb{P}_F$ is totally ramified in $L\slash F$, there is one more rational place of $L$ lying above it. 
Altogether, we conclude that 
$$ N_L = |Z|\cdot \ell +1, $$
hence $|Z|=\frac{N_L-1}{\ell}$,
and the proof follows.
\end{proof}

By \cref{cor:distance of trace and as}, to lower bound the distance of $\trcode$, it suffices to upper bound the ratio $\frac{N_L-1}{n} = \frac{N_L-1}{N_F-1}$, or, essentially equivalently, the more natural ratio $\frac{N_L}{N_F}$. Note that the trivial upper bound $N_L \le \ell n + 1$ yields only the trivial lower bound $\delta \ge 0$. Obtaining a nontrivial upper bound on $N_L$ is the content of \cref{sec:bound on nl}.

\section{Our Bound on the Number of Rational Points}\label{sec:bound on nl}
In this section, we prove the following fairly general—though somewhat technical to state—proposition, which serves as the foundation for our results on elementary abelian $p$-extensions and Kummer extensions. The statement of \cref{prop:general} is formulated in a broad setting and thus involves several unspecified parameters. In \cref{cor:still quite general}, we present new normalized parameters, and rewrite the bound with these notations. In \cref{prop:After beta and mu instantiation} we instantiate these parameters, and obtain the main general result, although still quite technical.
In \cref{sec:extensions we work with} we show that a quite general family of extensions satisfy the conditions of \cref{prop:After beta and mu instantiation}, and conclude \cref{thm:extensions we work with}, which will be used in all our applications.

From here onwards, $p$ is a prime number and $q = p^u$ for some integer $u \ge 1$. Moreover, $F/\FF_q$ is a function field with $N$ rational places and genus $g$. We focus on curves in the regime 
\begin{equation}\label{eq:g_omega_N_over_sqrt_q}
    g = \Omega\left ( \frac{N}{\sqrt{q}}\right ).
\end{equation}
For the complementary regime, one can apply the Grothendieck's trace formula to bound the number of rational places over extensions of \(F\); see the discussion in ~\cref{sec:intro trace of ag}.
Let $\p$ be a rational place of $F$, and let $x \in \rrs(s \p) \setminus \rrs((s-1)\p)$ for some integer $s > 1$. Hence $\deg (x)_\infty = s$. We consider field extensions of the form $L/F$ where $L = F(z)$. We denote by 
\[
t \deff \deg(z)_\infty.
\]
The reader may think of $t$ as given as input, and the resulted bound depends on $t$. The parameter $s$ on the other hand is ``internal'' to the proof and it will be beneficiary to choose it as small as possible.\footnote{Note that by Claim \ref{clm:LB_for_the_gonality}, $s\geq \frac{N}{q+1}$.}

We will need to introduce a few more parameters. Let $1 \le m \le q$ be a power of $p$. Denote 
\[
A \deff \left\lfloor \frac{Nm}{2q} \right\rfloor
\]
and let $B \in \mathbb{N}$ be a parameter.

\begin{proposition}\label[proposition]{prop:general}
	Let $L/F$ be a function field extension of the form $L = F(z)$, and let $d \deff [L:F]$.
	Assume that $\p$ is totally ramified in $L$, and denote by $\pp$ the unique place of $L$ lying above $\p$. Let $\{b_i\}_{i \in \mathcal{I}}$ be a basis of $\rrs(A\pp)$ such that $v_\pp(b_i) = -i$. Assume that the elements
	\begin{equation}\label{eq:S}
		S = \left\{ b_i x^{jm} z^{km} : i \leq A,\, j \le B,\, k < d \right\}
	\end{equation}
	are linearly independent over $\FF_q$. Further assume that
	\begin{ineq}\label{ineq:general prop assumption}
		(A - g_L)Bd > \frac{N}{2} + dsB + (d-1)t,
	\end{ineq}
	where $g_L$ denotes the genus of $L$. Then, the number of rational points $N_L$ of $L$ satisfies
	\[
	N_L \le \left( sdB + \frac{N}{2q} + td \right) m.
	\]
\end{proposition}

	\begin{proof}[Proof of \cref{prop:general}]
        First, notice that as $\p$ is totally ramified, $\FF_q$ is the full constant field of $L$.
	Let $U$ be the $\FF_q$ vector space spanned by $S$ as given in \cref{eq:S}.
	Consider the $\FF_p$ linear map
	$
	\Psi : U \to L
	$
	defined by
	$$
	\Psi\left( \sum_{\substack{i \leq A \\ j \le B, k < d}}{\!\!\!c_{i,j,k} b_ix^{jm}z^{km}} \right) = \sum_{i,j,k} c_{i,j,k}^{q/m} b_i^{q/m} x^j z^k,
	$$
	where $c_{i,j,k} \in \FF_q$. Note that
	$$
	\mathrm{Im}\Psi \subseteq \rrs\left(\left(N/2+dsB \right)\pp +(d-1) (z)_\infty\right),
	$$
	where we used that
	$
	\v_\pp(x) = e(\pp|\p) v_\p(x) = ds.
	$
        This, together with \cref{clm:UB_for_rrs_dim}, shows that
        \[
            \dim_{\FF_p} \mathrm{Im} \Psi \leq \left (N/2 + dsB + (d-1)t + 1 \right ) \cdot [\FF_q:\FF_p].
        \]
	Per our assumption on $S$, as an $\FF_q$-vector space,
	$$
	\dim_{\FF_p} U = \dim_{\FF_q} \rrs(A \pp)\cdot (B+1)\, d \cdot [\FF_q:\FF_p] \ge (A-g_L)\, (B+1)\, d \cdot [\FF_q:\FF_p],
	$$
	where the last inequality follows by Riemann's Theorem.
	
	Therefore our assumption \cref{ineq:general prop assumption} implies that $\dim_{\FF_p} U > \dim_{\FF_p} \mathrm{Im}\Psi$ and so there exists a nonzero element 
    $h \in U$ such that $\Psi(h) = 0$. Indeed, this follows since $\Psi$ is additive. Note however that for every degree-1 place $\qq \neq \pp$ of $L$,
	$$
	0 = \Psi(h)^m(\qq) = h(\qq).
	$$
	That is, $h$ vanishes on all rational places $\qq\neq\pp$ of $L$, and so $N_L \le \deg(h)_\infty + 1$. But
	$$
	h \in U \subseteq \rrs\left(\left(A+sdBm\right)\pp + (d-1)m (z)_\infty \right), 
	$$
	hence
	$
	\deg(h)_\infty \le A+sdBm+t(d-1)m,
	$
	which concludes the proof.
	\end{proof}
	
	We define the parameters $\tau,\beta,\gamma,\mu$ and $\sigma$ which satisfies the following relations with regards to the previously defined parameters in this section so far:
	\begin{align}\label{eq:normalized greeks}
		t &= \tau \frac{ N}{\sqrt{q}} \notag\\
		B &= \beta\frac{ \sqrt{q}}{d} \notag \\
		g &= \gamma\frac{ N}{\sqrt{q}} \notag\\
		m &= \mu \sqrt{q}  \notag\\
		s &= \sigma \frac{ N}{q}
	\end{align}
    Note that per \cref{eq:g_omega_N_over_sqrt_q} we assume that $\gamma = \Omega(1)$.
	
    With this we have the following corollary, retaining the notation from \cref{prop:general}.
	\begin{corollary}\label[corollary]{cor:still quite general}
		Assume that
		\begin{ineq}\label{ineq:general gl bound}
		g_L \le d(g+t)
		\end{ineq}
		and that
		\begin{ineq}\label{ineq:cor constraint dim}
			\beta \mu \ge 
			    1 + 2(\tau+\gamma)\beta d+
			    \frac2{\sqrt{q}} \left( \tau (d-1)+
			    \sigma \beta \right)
		\end{ineq}
		hold. Assume further, as in \cref{prop:general}, that the elements in $S$ are linearly independent over $\FF_q$. Then,
		\begin{ineq}\label{ineq:general bound on ratio}
		N_L \le \mu N \!\left(\sigma \beta + \tau d + \frac1{2\sqrt{q}} \right).
		\end{ineq}
	\end{corollary}
	
	\begin{proof}
		Per our assumption, \cref{ineq:general gl bound}, and using the parameters defined in \cref{eq:normalized greeks}, we have that the LHS of \cref{ineq:general prop assumption} can be bounded from below by
		\begin{align*}
			(A-g_L) Bd &\ge \left(\frac{Nm}{2q} - d (g+t) \right) Bd \\
			&= (\mu/2 - (\tau+\gamma)d)\beta N ,
		\end{align*}
		whereas the RHS of \cref{ineq:general prop assumption} can be rewritten as
		$$
		\frac{N}{2}+dsB+(d-1)t = \left( \frac{1}{2} + \frac1{\sqrt{q}} \left( \tau (d-1) + \sigma \beta  \right) \right)N.
		$$
		Thus, \cref{ineq:general prop assumption} in \cref{prop:general} holds per our assumption, \cref{ineq:cor constraint dim}. Thus, we can invoke \cref{prop:general} to conclude that
		\begin{align*}
			N_L &\le \left( sdB + \frac{N}{2q} + td \right)m \\
			&= \mu N \left(\sigma \beta + \tau d + \frac1{2\sqrt{q}} \right).
		\end{align*}
	\end{proof}
	
	\subsection{Before the Instantiation: A Parameter Walkthrough}\label{subsec:parameter walkthrough}
	Before proceeding to instantiate the parameters $\beta, \mu$ in \cref{cor:still quite general}, we provide the reader with some informal insight into the magnitudes of the various parameters. While this discussion is not rigorous, it is intended to offer intuition about the bounds one should expect.
	
	The five parameters appearing in \cref{eq:normalized greeks} can be grouped into three distinct categories. As previously noted, the reader should view $\tau$ as the input parameter -- the bound we derive for $N_L$ depends on $\tau$, and it is the nature of this dependence that we aim to understand. The parameters $\beta$ and $\mu$ are under our control; by selecting them appropriately, we can optimize the resulting bound. In contrast, the parameters $\gamma$ and $\sigma$ are intrinsic to the function field under consideration, and for the purposes of this discussion, the reader may assume $\gamma = \sigma = 1$.

As a starting point, we ignore terms that vanish as $q \to \infty$. This simplification allows us to gain an initial understanding of how the bound depends on the dominant quantities, namely, $p$, $d$ and $\tau$. It also offers guidance on how to optimize the bound by appropriately choosing the parameters $\beta$ and $\mu$. The assumption in \cref{ineq:cor constraint dim} then simplifies to
\[
\beta \mu \ge 1 + 2(\tau + 1)\beta d.
\]
We take this to hold with equality, yielding the following relation between $\beta$ and $\mu$:
\begin{equation}\label{eq:informal mu beta}
\mu = \frac{1}{\beta} + 2(\tau + 1)d.
\end{equation}
Now, the bound guaranteed by \cref{cor:still quite general} simplifies, as discussed above, to
\[
\frac{N_L}{N} \le \mu \left(\beta + \tau d\right).
\]
Plugging in the relation between $\mu$ and $\beta$, we obtain
\begin{align*}
	\frac{N_L}{N} &\le \left( \frac{1}{\beta} + 2(1 + \tau)d \right) \left(\beta + \tau d\right).
\end{align*}

Generally, for $a,b>0$, the minimum of the function $f(x) = (x + a)\left(\frac{1}{x} + b\right)$ in $(0,\infty)$ is $(1 + \sqrt{ab})^2$, attained at $x = \sqrt{\frac{a}{b}}$. Applying this, we get that setting $\beta = \sqrt{\frac{\tau }{2(1+\tau)}}$ yields the bound
\[
\frac{N_L}{N} \le \left(1 + \sqrt{2\tau(1 + \tau)}\, d\right)^2.
\]
Analogously to the Hasse--Weil theorem, where the number of rational points on a curve lying over the projective line is expressed as $q + 1 + \mathrm{Err}$, the presence of the constant term $1$ is natural and expected.
The remaining contribution is the ``error term''. Note that a bound of $d$ on the ratio is trivial, since $L/F$ is an extension of degree $d$. Thus, the bound is meaningful only when $\tau \ll 1$, in which case we can simplify the informal (and inaccurate) bound to
\[
\frac{N_L}{N} \le \left(1 + \sqrt{2\tau}\, d\right)^2 = 1 + 2\sqrt{2\tau}\, d + 2\tau d^2.
\]

Note that the contributions of the two summands on the right-hand side to the error term are incomparable; which one dominates depends on whether $2\tau < \frac{1}{d^2}$ or not.
\paragraph{The dependence on $q$.}
Of course, this bound is not accurate — in particular, it ignores the dependence on $q$, which is actually quite important in applications, as it determines how the alphabet reduction (i.e., the minimal value of $q$ for which we start) affects the bound. Moreover, the parameter $\beta$ should be chosen such that $B$ is an integer, and $\mu$ should be chosen such that $m$ is a power of $p$, which we will take care of later.

However, the above discussion provides useful insight into how to choose $\beta$ and $\mu$. In particular, we set
$
\beta = \sqrt{\frac{\tau}{2(1+\tau)}},
$
which is well-approximated by $\sqrt{\tau/2}$ under the assumption that $\tau \ll 1$. Thus, we proceed with the choice $\beta = \sqrt{\tau/2}$. Substituting this into \cref{eq:informal mu beta}, and using again that $\tau \ll 1$, a suitable choice for $\mu$ is
\begin{equation}\label{eq:choice of mu}
    \mu = \frac{1}{\sqrt{\tau/2}} + 2d.
\end{equation}
Plugging this into \cref{ineq:general bound on ratio}, we obtain
\begin{ineq}\label{ineq:informal_bound}
    \frac{N_L}{N} \le \left(1 + \sqrt{2\tau} d \right)^2 + \frac{1}{\sqrt{q}} \left( d + \frac{1}{\sqrt{2\tau}} \right).
\end{ineq}
Notice that by \cref{clm:LB_for_the_gonality}, we have $\tau = t \frac{\sqrt{q}}{N} \geq \frac{1}{\sqrt{q}} \cdot \frac{q}{q+1}$. Thus,
\begin{align*}
    \frac{d}{\sqrt{q}} &\leq \tau d \cdot \left ( 1+q^{-1}\right ), \\
    \frac{1}{\sqrt{2q\tau}} &\leq \sqrt{\tau}.
\end{align*}
Hence, the bound in \cref{ineq:informal_bound} can be rewritten as
\begin{ineq}\label{ineq:bound O notation}
    \frac{N_L}{N} \leq 1 + O(\sqrt{\tau}d + \tau d^2).
\end{ineq}

\subsection{Instantiation of Parameters}
In this subsection, we formalize the above discussion and establish the following general result.
\begin{proposition}\label[proposition]{prop:After beta and mu instantiation}
    Let $L/F$ be a function field extension of the form $L = F(z)$, and let $d \deff [L:F]$.
	Assume that $\p$ is totally ramified in $L$, and denote by $\pp$ the unique place of $L$ lying above $\p$. Let $\{b_i\}_{i \in \mathcal{I}}$ be a basis of $\rrs(A\pp)$ such that $v_\pp(b_i) = -i$. Assume that the elements
	\begin{equation}
		S = \left\{ b_i x^{jm} z^{km} : i \leq A,\, j \le B,\, k < d \right\}
	\end{equation}
	are linearly independent over $\FF_q$. Further, assume that $g_L \le d(g+t)$.
    Then,
    \begin{equation}
        \frac{N_L}{N} \leq \sigma + O_{\sigma}\left ( \sqrt{\tau} d \gamma p + \tau d^2 \gamma p + \frac{d}{\sqrt{ q}} \left ( \frac{1}{\sqrt{\tau}} + d \right ) \gamma p \right ).
    \end{equation}
\end{proposition}

\begin{remark}
    If $\sigma \sim 1, \gamma \sim 1$ as in \cref{subsec:parameter walkthrough}, this is the same error term as in \cref{ineq:informal_bound} up to a factor of $p$.
\end{remark}

\begin{proof}
    By \cref{cor:still quite general}, if
    \begin{ineq}\label{ineq:beta inequality Artin-Schreier}
        \beta \mu \ge 
			    1 + 2(\tau+\gamma)\beta d+
			    \frac2{\sqrt{q}} \left( \tau (d-1)+
			    \sigma \beta \right),
    \end{ineq}
    then,
    \begin{ineq}\label{ineq:internal to proof upper bound before instantiation}
		N_L \le \mu N \left(\sigma \beta + \tau d + \frac1{2\sqrt{q}} \right).
    \end{ineq}
    
    We now choose values for the parameters $\beta$ and $\mu$. Similarly to \cref{subsec:parameter walkthrough}, we want to take \cref{ineq:beta inequality Artin-Schreier} with equality. But we have to make sure that $B = \beta \frac{\sqrt{q}}{d}$ is an integer. Isolating $\beta$ in \cref{ineq:beta inequality Artin-Schreier}, as long as $\mu > 2d(\tau + \gamma) + \frac{2\sigma}{\sqrt{q}}$, we get
\begin{ineq}
    \beta \geq \frac{1 + \frac{2}{\sqrt{q}}\tau(d-1)}{\mu - 2d(\tau + \gamma) - \frac{2\sigma}{\sqrt{q}}}.
\end{ineq}
To ensure that $B$ is an integer, we can choose
\[
    \beta = \frac{1 + \frac{2}{\sqrt{q}}\tau(d-1)}{\mu - 2d(\tau + \gamma) - \frac{2\sigma}{\sqrt{q}}} + O\left(\frac{d}{\sqrt{q}}\right ).
\]
Notice the inequality $\frac{1}{1-x} \leq 1 + 2x$ which holds for $0 \leq x \leq 0.5$. Under the assumption 
\begin{ineq}\label{ineq:assumption_on_mu}
    \frac{d}{\mu}(\tau + \gamma) + \frac{\sigma}{\sqrt{q}\mu} \leq \frac{1}{4},
\end{ineq}
we have
\[
    \beta \leq \frac{1}{\mu} \left ( 1 + \frac{2}{\sqrt{q}}\tau(d-1) \right )
    \left ( 1 + \frac{4d}{\mu}(\tau + \gamma) + \frac{4\sigma}{\mu \sqrt{q}} \right ) + O\left ( \frac{d}{\sqrt{q}} \right ).
\]
We want to choose $\mu$ similarly to \cref{eq:choice of mu}, but we have to make sure that $m$ is a power of $p$, and that \cref{ineq:assumption_on_mu} is satisfied. Pick $0\leq \alpha \leq 1$ and $C = C(\sigma)$ such that
\[
    \mu = \left (\frac{1}{\sqrt{\tau/2}} + 2d\right )\gamma C p^\alpha
\]
satisfies \cref{ineq:assumption_on_mu}, and such that $m$ is a power of $p$.
Instantiating those choices into \cref{ineq:internal to proof upper bound before instantiation}, we obtain
\begin{align*}
    \frac{N_L}{N} &\leq \mu \beta \sigma + \mu \left ( \tau d + \frac{1}{2\sqrt{q}} \right ) \\
    &\leq \sigma \left ( 1 + \frac{2}{\sqrt{q}}\tau(d-1) \right )
    \left ( 1 + \frac{4d}{\mu}(\tau + \gamma) + \frac{4\sigma}{\mu \sqrt{q}} \right ) + O\left ( \frac{\sigma \mu d}{\sqrt{q}} \right ) \\
    &\quad + \left(\tau d + \frac{1}{2\sqrt{q}}\right) \left (\frac{1}{\sqrt{\tau/2}} + 2d\right )\gamma C p^{\alpha}.
\end{align*}
Expanding using \cref{ineq:assumption_on_mu}, we have
\begin{align*}
    \frac{N_L}{N} &\leq \sigma \left ( 1 + \frac{4d}{\mu}(\tau + \gamma) + \frac{4\sigma}{\mu \sqrt{q}} + O \left ( \frac{\mu d}{\sqrt{q}} \right ) \right ) \\
    &\quad + C\sqrt{2\tau} d \gamma p^{\alpha} + C2\tau d^2 \gamma p^{\alpha} + C\frac{\gamma p^{\alpha}}{\sqrt{2q \tau}} + C\frac{d \gamma p^{\alpha}}{\sqrt{q}}.
\end{align*}
Now, notice that $\frac{1}{\mu} \leq \frac{\sqrt{\tau/2}}{C \gamma p^\alpha}$, and $p^\alpha \leq p$, hence
\begin{align*}
    C\sqrt{2\tau} d \gamma p^{\alpha} + C2\tau d^2 \gamma p^{\alpha} &+ C\frac{\gamma p^{\alpha}}{\sqrt{2q \tau}} + C\frac{d \gamma p^{\alpha}}{\sqrt{q}} \\
    &= O \left ( \sqrt{\tau} d \gamma p + \tau d^2 \gamma p + \frac{\mu d}{\sqrt{q}} \right ).
\end{align*}
If $\tau \geq 1$, the bound is trivial, therefore we can assume $\tau < 1$. As $\gamma = \Omega(1)$, we have
$
    \tau + \gamma = O(\gamma)
$, hence
\[
\frac{4d}{\mu}(\tau + \gamma) + \frac{4\sigma}{\mu \sqrt{q}} = O_\sigma(d\sqrt{\tau}).
\]
This implies
\begin{align*}
    \frac{N_L}{N} &\leq \sigma + O_\sigma \left ( \sqrt{\tau} d \gamma p + \tau d^2 \gamma p + \frac{\mu d}{\sqrt{q}} \right ) \\
    &= \sigma + O_{\sigma} \left (\sqrt{\tau} d \gamma p + \tau d^2 \gamma p + \frac{d}{\sqrt{ q}} \left ( \frac{1}{\sqrt{\tau}} + d \right ) \gamma p \right ).
\end{align*}
\end{proof}

\subsection{Extensions that Satisfy the Conditions of \cref{prop:After beta and mu instantiation}}\label{sec:extensions we work with}
\begin{proposition}\label[proposition]{prop:deg(f)=t and independence of S}
    Let $\p \in \mathbb{P}_F$ be a rational place, and let $f \in F$ have poles only at $\p$. Let $L/F$ be a function field extension defined by
	\[
	L = F(z) \quad \text{such that} \quad \phi(z) = f,
	\]
	where $\phi(T) \in \FF_q[T]$ is a polynomial of degree $d$. Assume that the polynomial $\phi(T) - f \in F[T]$ is irreducible. Assume that $\p$ is totally ramified in $L$, and denote by $\pp$ the unique place of $L$ lying above $\p$. Let $\{b_i\}_{i \in \mathcal{I}}$ be a basis of $\rrs(A\pp)$ such that $v_\pp(b_i) = -i$. Then,
    \begin{enumerate}
        \item $t \deff \deg_L(z)_\infty = -\v_\pp(z) = \deg_F(f)_\infty.$\label{claim_part:z_degree}
        \item If we can write $\deg_F (f)_\infty = t = \ell s + \varepsilon \frac{N}{q}$ for some $\ell$ relatively prime to $d$, and $|\varepsilon| < \frac{\sigma - 1/2}{d-1}$, then the set 
        \begin{equation}
		      S = \left\{ b_i x^{jm} z^{km} : i \leq A,\, j \le B,\, k < d \right\}
	    \end{equation}
        is linearly independent over $\FF_q$.\label{claim_item:independent_S}
    \end{enumerate}

\end{proposition}
\begin{proof}
    First, we prove that for every $h \in F$ we have
		$
		\deg_L(h)_\infty = d \deg_F (h)_\infty.
		$
    Indeed, since the place $\p$ is totally ramified, $\FF_q$ is the constant field of $L$ and so by Proposition 3.1.9 and Corollary 3.1.14 in~\cite{Stich} we have
		$$
		\deg_L(h)_\infty = \deg \con_{L/F} (h)_\infty = [L:F] \deg_F(h)_\infty = d \deg_F(h)_\infty.
		$$
    In particular, for $h=f$ we obtain
		$$
		d \deg_F(f)_\infty = \deg_L (f)_\infty = \deg_L(\phi(z))_\infty = d \deg_L(z)_\infty.
		$$
		We turn to prove the remaining equality appearing in \cref{claim_part:z_degree}, namely, $\deg_L(z)_\infty = -\v_\pp(z)$. As $f \in \rrs(r \p)$ it has a pole only at $\p$, and so $\v_\p(f) = -\deg_F(f)_\infty$. Consider now a pole $\qq$ of $z$. Then, $\v_\qq(z) < 0$ and so by the strict triangle inequality
		$
		\v_\qq(\phi(z)) = d \v_\qq(z) < 0.
		$
		But $\v_\qq(\phi(z)) = \v_\qq(f)$, and so $\qq$ lies over a pole of $f$. As the only pole of $f$ is $\p$ and since $\pp$ is the only place lying over $\p$, we get that $\qq = \pp$. Thus, $(z)_\infty = -\v_\pp(z) \pp$, and so
		$$
		\deg_L(z)_\infty = -\v_\pp(z) \deg \pp = -\v_\pp(z),
		$$
		where the last equality follows since $\pp\,|\,\p$ totally ramifies, hence $f(\pp|\p) = 1$ and $\deg \pp = f(\pp|\p) \deg \p = 1$. \\
        
    Next, we prove \cref{claim_item:independent_S}.
    Recall that $A = \lfloor \frac{Nm}{2q} \rfloor$. We claim that for 
    $i, i'\leq A,\  j, j' \leq B$ and $k,k' < d$ we have $(i, j, k) = (i', j', k')$ if and only if $v_\pp(b_i x^{jm} z^{km}) = v_\pp(b_{i'} x^{j'm} z^{k'm})$. This will complete the proof.
    
    Assume that $v_\pp(b_i x^{jm} z^{km}) = v_\pp(b_{i'} x^{j'm} z^{k'm})$. Denote $\Delta_i = i - i'$ and similarly for $j, k$.
    We have
    \begin{align*}
        0 = v_\pp(b_i x^{jm} z^{km}) - v_\pp(b_{i'} x^{j'm} z^{k'm}) = \Delta_i + \frac{mN}{q}\left(\sigma d \Delta_j + t \frac{q}{N} \Delta_k\right).
    \end{align*}
    Since $t = \ell s + \varepsilon \frac{N}{q}$, we can write $t \frac{q}{N} = \ell \sigma + \varepsilon$, to get
    \begin{align*}
        0 &= \Delta_i + \frac{mN}{q}(\sigma d \Delta_j + (\ell \sigma + \varepsilon) \Delta_k) \\
        &= \Delta_i + \frac{mN}{q}(\sigma(d\Delta_j + \ell \Delta_k) + \varepsilon \Delta_k),
    \end{align*}
    which is equivalent to
    \begin{equation}\label{eq:delta_i}
        -\Delta_i = \frac{mN}{q} \left(\sigma(d\Delta_j + \ell \Delta_k) + \varepsilon \Delta_k\right),
    \end{equation}
    Since $\gcd(\ell, d) = 1$ and $|\Delta_k|<d$,
    we have either $\Delta_k = 0$ or $\ell \Delta_k \not\equiv 0 \mod d$. 
    However, if $\ell \Delta_k \not\equiv 0 \mod d$, then
    \begin{align*}
        |\sigma (d \Delta_j + \ell \Delta_k)  + \varepsilon \Delta_k| \geq \sigma - (d-1) \varepsilon > \frac{1}{2}.
    \end{align*}
    Plugging this into \cref{eq:delta_i}, we get
    \[
        \frac{mN}{2q} < \left | \frac{mN}{q}\left(\sigma(d\Delta_j + \ell \Delta_k) + \varepsilon \Delta_k\right) \right | = |\Delta_i| \leq \frac{mN}{2q},
    \]
    which is a contradiction.
    Thus, $\Delta_k = 0$. 
    Now, if $\Delta_j \neq 0$, then by \cref{eq:delta_i} we have
    \[
        \frac{mN}{q} < \left | \frac{mN}{q}\sigma d \Delta_j \right | = |\Delta_i| \leq \frac{mN}{2q},
    \]
    which is again a contradiction. Thus $\Delta_j = \Delta_k = 0$, and hence $\Delta_i = 0$ which completes the proof.
\end{proof}

As a corollary, we state the main theorem which will be used in all our applications.
\begin{theorem}\label{thm:extensions we work with}
    Let $\p \in \mathbb{P}_F$ be a rational place, and let $f \in F$ have poles only at $\p$. Assume that $\deg_F(f)_\infty$ is relatively prime to $d$, and that we can write $\deg_F (f)_\infty = \ell s + \varepsilon \frac{N}{q}$ for some $\ell$ relatively prime to $d$, and $0 \leq \varepsilon < \frac{\sigma - 1/2}{d-1}$. Let $L/F$ be a function field extension defined by
	\[
	L = F(z) \quad \text{such that} \quad \phi(z) = f,
	\]
	where $\phi(T) \in \FF_q[T]$ is a polynomial of degree $d$. Assume that the polynomial $\phi(T) - f \in F[T]$ is irreducible. Assume that $\p$ is totally ramified in $L$. Further, assume that $g_L \le d(g+t)$.
    Then $\deg_F(f)_\infty = t$, and
    \begin{ineq}
        \frac{N_L}{N} \leq \sigma + O_{\sigma}\left ( \sqrt{\tau} d \gamma p + \tau d^2 \gamma p \right ).
\end{ineq}
\end{theorem}
\begin{proof}
    Let $\mathfrak{P} \in \mathbb{P}_L$ be a place lying above $\p$.
    The equation $\phi(z) = f$ implies that
    \[
        e(\mathfrak{P}|\p) \cdot\deg_F(f)_\infty = \deg_L(f)_\infty = d\, \deg_L(z)_\infty = d t.
    \]
    As $\deg_F(f)_\infty$ is relatively prime to $d$, we must have $\deg_F (f)_\infty = t$ and $e(\mathfrak{P}|\p) = d$. Hence we are in a position to apply \cref{prop:After beta and mu instantiation}, to obtain
    \begin{equation*}
        \frac{N_L}{N} \leq \sigma + O_{\sigma}\left ( \sqrt{\tau} d \gamma p + \tau d^2 \gamma p + \frac{d}{\sqrt{ q}} \left ( \frac{1}{\sqrt{\tau}} + d \right ) \gamma \right ).
    \end{equation*}
    To complete the proof, we use \cref{clm:LB_for_the_gonality} to observe that 
    $$
    \tau = t \frac{\sqrt{q}}{N} = \deg_F(f)_\infty \cdot \frac{\sqrt{q}}{N} \geq \frac{1}{2\sqrt{q}}.
    $$
    Hence
\begin{align*}
    \frac{d^2}{\sqrt{q}} &\leq 2 \tau d^2, \\
    \frac{d}{\sqrt{q\tau}} &\leq 2 \sqrt{\tau} d.
\end{align*}
\end{proof}

\section{Abelian Extensions and Character Sums}\label{sec:abelian_extensions}
In this section, we instantiate the bound from \cref{thm:extensions we work with} for two families of abelian extensions. 
The first family consists of elementary abelian $p$-extensions of the form $L = F(z)$, where $z^\ell - z = f$ and $d$ is a prime power. We apply the bound to this setting to derive a lower bound on the minimum distance of TAG codes, as discussed in \cref{sec:trace codes}. Moreover, in the Artin--Schreier case (i.e., when $\ell = p$ is prime), the bound yields an upper bound on exponential sums of functions over curves. The second family we consider is Kummer extensions, namely those of the form $L = F(z)$ with $z^d = f$, where $d$ is not divisible by the characteristic $p$ of $F$.

\subsection{Elementary Abelian $p$-Extensions and Exponential Sums}\label{subsec:p-abelian extensions}
In this subsection we consider elementary abelian $p$-extensions of the form $L = F(z)$, where $z^\ell - z = f$ and $\ell$ is a prime power. These extensions satisfy the conditions of \cref{thm:extensions we work with}.
As discussed in \cref{sec:trace codes}, any non-trivial bound for the number of rational points on this type of extensions immediately yields a lower bound on the minimum distance of the corresponding TAG code. In \cref{prop:Artin-Schreier bound before instantiation}, we use this connection to obtain a lower bound for the distance of TAG codes; in \cref{cor: distance trace code} we give an argument that also yields an \emph{upper} bound. We use this in  \cref{thm:exponential_sums} to get an upper bound for exponential sums of function on curves.

Recall the setting of \cref{sec:bound on nl}. Let $\ell>1$ be a power of $p$. Let $\phi(T) = T^{\ell} - T$, let $f \in \rrs(r\p)$ and assume that $t = \deg(f)_\infty = -v_\p(f)$ is not divisible by $p$. Let $L = F(z)$ with $\phi(z) = f$. By  \cite[Proposition 3.7.10(a)]{Stich},  $L / F$ is an elementary abelian extension of exponent $p$.
By part (d) of this proposition, the prime $\p$ is totally ramified in $L$, and moreover by part (e) we have
\[
    g_L = \ell g + \frac{\ell-1}{2} ( -2 + (t + 1)) \leq \ell (g + t).
\]
Thus we are in a position to apply \cref{thm:extensions we work with} to $L / F$.
\begin{proposition}\label{prop:Artin-Schreier bound before instantiation}
    In the above setting, assume further that $f$ satisfies the condition in \cref{claim_item:independent_S} of \cref{prop:deg(f)=t and independence of S}. Then,
    \begin{ineq}
        \frac{N_L}{N} \leq \sigma + O_{\sigma}\left ( \ell \sqrt{\tau} \gamma p + \ell^2 \tau \gamma p \right ).
    \end{ineq}
\end{proposition}

Let $\tr$ be the trace function from $\FF_q$ to $\FF_\ell$. 
\cref{c:splits whenever trace vanishes} relates splitting places in $L / F$ to the vanishing of the trace of $f$ at these places. We use this relation together with \cref{prop:Artin-Schreier bound before instantiation} to get an upper bound for the number of vanishing traces of evaluations of $f$.
\begin{corollary}\label[corollary]{cor: vanishing trace upper bound}
    For every function $f \in \rrs(r\p)$ such that $t = \deg (f)_\infty$ is not divisible by $p$, we have
    \begin{equation*}
        \left|\left\{\mathfrak{q} \in \mathbb{P}^1_F \setminus \{\p\}: \tr(f(\mathfrak{q})) = 0 \right\}\right| = \frac{N_L - 1}{\ell}.
    \end{equation*}
    If moreover $f$ satisfies the condition in \cref{claim_item:independent_S} of \cref{prop:deg(f)=t and independence of S}, then
    \begin{ineq}\label{ineq: upper bound on trace code weight}
        \frac{1}{N - 1}\ \left|\left\{\mathfrak{q} \in \mathbb{P}^1_F \setminus \{\p\}: \tr(f(\mathfrak{q})) = 0 \right\}\right| \leq \frac{\sigma}{\ell} +  O_{\sigma} \left ( p \sqrt{\tau} \gamma + p \ell \tau \gamma\right ).
    \end{ineq}
\end{corollary}

As we will see in \cref{sec:specific function fields}, for all function fields in which we instantiate our results, it will be possible to choose $x \in F$ such that $\sigma \leq 1$. Moreover, \cref{prop:deg(f)=t and independence of S} provides a condition for the set $S = S(f)$ to be linearly independent over $\FF_q$, depending only on $t$. Consequently, if the result applies to $f$, it also applies to $f - \alpha$ for all $\alpha \in \FF_q$. \\ In this case, we can also obtain a lower bound.
\begin{corollary}\label[corollary]{cor: distance trace code}
    Assume that $\sigma \leq 1.$ For every function $f \in \rrs(r\p)$ of degree not divisible by $p$, and that satisfies the condition in \cref{claim_item:independent_S} of \cref{prop:deg(f)=t and independence of S}, we have
    \begin{equation*}
        \frac{1}{N - 1}\ \left|\left\{\mathfrak{q} \in \mathbb{P}^1_F \setminus \{\p\}: \tr(f(\mathfrak{q})) = \beta \right\}\right| = \frac{1}{\ell} +  O \left ( \ell \sqrt{\tau} \gamma p \right ).
    \end{equation*}
\end{corollary}
\begin{proof}
    Let $\beta \in \FF_\ell$. Since the map $\tr: \FF_q \to \FF_\ell$ is onto, there exists some $\gamma \in \FF_q$ such that $\tr(\gamma) = \beta$.
    Applying \cref{cor: vanishing trace upper bound} to $f - \gamma$, we obtain
\[
    \frac{1}{N - 1}\ \left|\left\{\mathfrak{q} \in \mathbb{P}^1_F \setminus \{\p\}: \tr(f(\mathfrak{q})) = \beta \right\}\right| \leq \frac{1}{\ell} +   O\left ( \sqrt{\tau} \gamma p + \ell \tau \gamma p \right ).
\]
    For $\alpha \in \FF_\ell$, let
\[
    \mathcal{S}_\alpha = \left|\left\{\mathfrak{q} \in \mathbb{P}^1_F \setminus \{\p\}: \tr(f(\mathfrak{q})) = \alpha \right\}\right|.
\]
Since
\[
    \sum_{\alpha \in \FF_\ell} \mathcal{S}_\alpha = N - 1,
\]
we get
\[
    1 - \frac{\mathcal{S}_\beta}{N - 1} = \frac{1}{N-1}\sum_{\beta \neq \alpha \in \FF_\ell} \mathcal{S}_\alpha \leq 1 - \frac{1}{\ell} + O \left ( \ell \sqrt{\tau} \gamma p + \ell^2 \tau \gamma p \right ).
\]
Thus,
\[
    \frac{1}{\ell} - O \left ( \ell \sqrt{\tau} \gamma p + \ell^2 \tau \gamma p \right ) \leq \frac{\mathcal{S}_\beta}{N-1}.
\]
To complete the proof, notice that if $\sqrt{\tau}\ell \geq 1$ the result is trivial, and otherwise $\tau \ell^2<\sqrt{\tau}\ell < 1$.
\end{proof}
In \cref{sec:specific function fields} we use these results to construct and analyze TAG codes on some curves.

Lastly, from the special case of Artin-Schreier extensions, i.e. when $\ell=p$, we conclude an upper bound for the exponential sums arising from these functions.
\begin{theorem}\label{thm:exponential_sums}
    Let $\ell = p$. Assume that $\sigma \leq 1.$ For every function $f \in \rrs(r\p)$ of degree not divisible by $p$, and that satisfies the condition in \cref{claim_item:independent_S} of \cref{prop:deg(f)=t and independence of S}, we have
\[
    \frac{1}{N}\,\Bigg|\sum_{\q \in \mathbb{P}^1_F \setminus \{ \p \} } e^{\frac{2\pi i}{p}\tr ( f (\q))}\Bigg| = O\left (\sqrt{\tau} \gamma p^3 \right ).
\]
\end{theorem}

\begin{proof}
    Applying \cref{cor: distance trace code} to the case $\ell = p$, we get
    \begin{align*}
        \frac{1}{N}\,\Bigg|\sum_{\q \in \mathbb{P}^1_F \setminus \{ \p \} } e^{\frac{2\pi i}{p} \tr ( f (\q))}\Bigg| &= \frac{1}{N} \,\Bigg|\sum_{\beta = 0}^{p-1} e^{\frac{2\pi i}{p} \beta} \left|\left\{\mathfrak{q} \in \mathbb{P}^1_F \setminus \{\p\}: \tr(f(\mathfrak{q})) = \beta \right\}\right|\Bigg| \\
        &\leq \frac{1}{N}\sum_{\beta=0}^{p-1}\left|\left\{\mathfrak{q} \in \mathbb{P}^1_F \setminus \{\p\}: \tr(f(\mathfrak{q})) = \beta \right\}\right|
        \\
        &\leq \sum_{\beta=0}^{p-1} \left ( \frac{1}{p} +  O \left (\sqrt{\tau} \gamma p^2 \right ) \right ) \\
        &= O \left (\sqrt{\tau} \gamma p^3 \right ),
    \end{align*}
    which finishes the proof.
\end{proof}

\subsection{Kummer Extensions and Multiplicative Character Sums}\label{sec:Kummer}

In this section, following a similar outline to \cref{subsec:p-abelian extensions}, we instantiate the bound from \cref{thm:extensions we work with} in the setting of Kummer extensions to get an upper bound on the number of rational places in such extensions. Similarly to \cref{cor: distance trace code}, in \cref{prop:kummer equality} we give an argument that also yields a \emph{lower} bound. We then apply this bound to derive upper bounds for multiplicative character sums of functions on curves, in direct analogy with the upper bounds for exponential sums obtained for the Artin–Schreier case.

Recall the setting of \cref{sec:bound on nl}. Let $\phi(T) = T^{d}$ for $d>1$ such that $d \mid q-1$. Let $f \in \rrs(r\p)$ and assume that $t = \deg(f)_\infty = -v_\p(f)$ is relatively prime to $d$. Let $L = F(z)$ with $\phi(z) = f$. 
Since $d\mid q-1$, the field $\FF_q$ contains a primitive $d$-th root of unity.
By Corollary 3.7.4 in \cite{Stich}, $L / F$ is a cyclic Kummer extension of degree $d$, and hence $\phi(T) - f$ is irreducible.
By Proposition 3.7.3(b) in \cite{Stich}, the place $\p$ is totally ramified in $L$, and moreover by Corollary 3.7.4 in~\cite{Stich} we have
\[
    g_L = 1 + d(g-1) + \frac{1}{2}\sum_{\q \in \mathbb{P}_F}\left(d-(d, v_\q(f))\right) \deg \q.
\]
Note that for $\q\in\mathbb{P}_F$, if $v_\q(f)=0$ then $d-(d,v_\q(f))=0$. Otherwise, $d-(d,v_\q(f))\leq d$. Therefore,
\begin{align*}
    g_L &\leq 1+d(g-1) + \frac{1}{2}\cdot\sum_{\q\in\mathbb{P}_F\,:\,v_\q(f)\ne 0}d\deg\q 
    \\
    &\leq 
    dg + \frac{d}{2}\left(\deg(f)_0 + \deg(f)_\infty\right)\leq d(g+t)
\end{align*}
as $\deg(f)_0=\deg(f)_\infty=t$. 
Thus we are in a position to apply \cref{thm:extensions we work with} to $L / F$.
\begin{proposition}\label[proposition]{prop:Kummer bound before instantiation}
    In the above setting, assume further that $\deg_F(f)_\infty$ is relatively prime to $d$, and satisfies the condition in \cref{claim_item:independent_S} of \cref{prop:deg(f)=t and independence of S}. Then,
    \begin{ineq}
        \frac{N_L}{N} \leq \sigma + O_{\sigma}\left ( \sqrt{\tau} d \gamma p + \tau d^2 \gamma p \right ).
    \end{ineq}
\end{proposition}

Similarly to \cref{subsec:p-abelian extensions}, we can also achieve a lower bound under further assumptions.

\begin{proposition}\label[proposition]{prop:kummer equality}
     Assume that $\sigma \leq 1$. For every function $f \in \rrs(r\p)$ such that $\deg_F(f)_\infty$ is relatively prime to $d$, and satisfies the condition in \cref{claim_item:independent_S} of \cref{prop:deg(f)=t and independence of S}, we have
    \begin{equation*}
        \frac{N_L}{N} = 1 + O\left ( \sqrt{\tau} d^2 \gamma p \right ).
    \end{equation*}
    Furthermore, we have
    \[
        \frac{1}{N} \left| \left\{ \q \in \mathbb{P}^1_F \setminus \{ \p \} : \exists\,y \in \FF_q^\times,\ y^d = f(\q ) \right\}\right| = \frac{1}{d} + O\left ( \sqrt{\tau} d \gamma p \right ).
    \]
\end{proposition}

\begin{proof}
    Fix representatives $\{ \epsilon_1 = 1, \ldots , \epsilon_d \}$ of $\left ( \FF_q^\times / \FF_q^{d\times} \right )$. For $i=1,\ldots,d$, let
    \[
        \mathcal{S}_i = \left\{ \q \in \mathbb{P}^1_F \setminus \{ \p \} : \exists\,y \in \FF_q^\times,\ y^d = \epsilon_i f(\q ) \right\}.
    \]
    Note that the sets $\mathcal{S}_i$ form a partition of the set $\{\q \in \mathbb{P}^1_F \setminus \{\p\} : f(\q)\neq 0 \}$.
    Let
    \[
        L_i := F(z_i),\ z_i^d = \epsilon_i f.
    \]
    Since $L_i/F$ is Galois, every $\q \in \mathbb{P}_F$ decomposes in $L_i$ to places of the same degree with the same ramification index. Hence, rational places in $L_i$ must lie above splitting rational places in $F$, or above ramified rational places. Let $\mathbb{P}_{L_i}^1$ be the set of rational place in $L_i$. We have,
    \begin{align*}
        N_{L_i} &= \left| \left\{\mathfrak{Q} \in \mathbb{P}_{L_i}^1:\mathfrak{Q} \text{ lies above }\q \text{ such that } f(\q) \neq 0 \right\}\right|\\ 
        &+ \left| \left\{\mathfrak{Q} \in \mathbb{P}_{L_i}^1:\mathfrak{Q} \text{ lies above }\q \text{ such that } f(\q) = 0 \right\}\right| \\
        &+ 1,
    \end{align*}
    where the last term corresponds to the totally ramified place lying above $\p$.
    By \cref{thm:kummers theorem}, if $f(\q) \neq 0$, then $\q \in \mathcal{S}_i$ if and only if $\q$ splits completely into $d$ rational places in $L_i$. Hence the first summand in the RHS is $d|\mathcal{S}_i|$.
    As for the second summand, note that the number of rational places $\q$ in $F$ with $f(\q) = 0$ is at most $\tau \frac{N}{\sqrt{q}} = t = \deg(f)_0$, and there are at most $d$ places in $L_i$ lying above each. Thus,
    \begin{equation}\label{eq:difference N_(L_i) and S_i}
        0 \leq N_{L_i} - d |\mathcal{S}_i| = O \left ( d \tau \frac{N}{\sqrt{q}} \right ).
    \end{equation}
    By applying \cref{prop:Kummer bound before instantiation} to $L_i$ in the inequality above on the left, we obtain
    \begin{ineq}\label{ineq:S_i bound}
        \frac{|\mathcal{S}_i|}{N-1} \leq \frac{1}{d} + O\left ( \sqrt{\tau} \gamma p + \tau d \gamma p \right ).
    \end{ineq}
    On the other hand we have
    \[
        \sum_{i = 1}^d |\mathcal{S}_i| + \left| \left\{ \q \in \mathbb{P}^1_F \setminus \{ \p \} : f(\q) = 0 \right\}\right| = N-1. 
    \]
    Thus,
    \begin{align*}
        \frac{N_L}{N-1} &\geq \frac{d}{N-1}|\mathcal{S}_1| \\ &\geq  d - \frac{d}{N-1}\left| \left\{ \q \in \mathbb{P}^1_F \setminus \{ \p \} : f(\q) = 0 \right\}\right| - \frac{d}{N-1}\sum_{i = 2}^d |\mathcal{S}_i|.
    \end{align*}
Notice that $\tau \leq 1$, and hence $\deg (f)_0 \leq \frac{N}{\sqrt{q}}$. Therefore,
\[
    \left| \left\{ \q \in \mathbb{P}^1_F \setminus \{ \p \} : f(\q) = 0 \right\}\right|  = O\left(\frac{N}{\sqrt{q}} \right ).
\]
Combined with \cref{ineq:S_i bound}, we obtain
\[
    \frac{N_L}{N-1} \geq d - O \left ( \frac{d}{\sqrt{q}} \right ) - (d-1) - O\left ( \sqrt{\tau} d^2 \gamma p + \tau d^3 \gamma p \right ).
\]
To complete the proof of the first part, notice that if $\sqrt{\tau}d \geq 1$ the result is trivial, and otherwise $\tau d^2<\sqrt{\tau}d < 1$. \\
To conclude the second part, plug the first part into the right equality in \cref{eq:difference N_(L_i) and S_i}.
\end{proof}
Lastly, we conclude the bound for multiplicative character sums arising from these functions.
\begin{theorem}
    For all multiplicative characters $\chi \in \widehat{\FF_q^\times}$ of order $d$, and for all functions $f\in F$ which satisfy the conditions of \cref{prop:kummer equality}, we have
    \[
        \frac{1}{N}\,\left|\sum_{\q \in \mathbb{P}^1_F \setminus \{\p \}} \chi(f(\q))\right| =
        O\left ( \sqrt{\tau} d^2 \gamma p \right ).
    \]
\end{theorem}
\begin{proof}
We use the notation of $\mathcal{S}_i$ and $\epsilon_i$ used in \cref{prop:kummer equality}.
    Notice that,
    \begin{align*}
        \frac{1}{N} \sum_{\q \in \mathbb{P}^1_F \setminus \{\p \}} \chi(f(\q)) &=
        \frac{1}{N}\sum_{i=1}^d \chi(\epsilon_i)^{-1} \mathcal{S}_i \\ 
        &= \left ( \frac{1}{N}\sum_{i=1}^d \chi(\epsilon_i)^{-1} \frac{1}{d} \right ) +  O\left ( \sqrt{\tau} d^2 \gamma p \right ) \\ 
        &=  O\left ( \sqrt{\tau} d^2 \gamma p \right ),
    \end{align*}
    where we have used the second part of \cref{prop:kummer equality}.
\end{proof}

\section{TAG Codes Instantiations}\label{sec:specific function fields}
In this section, we use the results obtained in the previous sections to construct and analyze TAG codes of specific function fields $F$. As discussed in \cref{subsubsec:constant rate TAGs}, the functions captured by our result are roughly those of degree $t < \frac{g}{p^2}$. Therefore, in order to construct TAG codes with high rate, it is necessary to consider curves that admit a place for which the associated Riemann-Roch spaces of degree $< \frac{g}{p^2}$ have sufficiently large dimensions.

\subsection{The Hermitian TAG Code}
Let $r$ be a power of a prime $p$, and let $q = r^2$. The Hermitian function field over $\FF_q$ is defined by
\begin{equation}\label{eq:hermitian_curve_eq}
    F = \FF_q(x, y), \qquad y^r + y = x^{r+1}.
\end{equation}
It is an elementary abelian $p$-extension of $\FF_q(x)$ of degree $r$.
Consider the place $\p_\infty \in \mathbb{P}_{\FF_q(x)}$. Let $\mathfrak{P}$ be some place lying above it. Then
\[
v_\pp(x^{r+1}) = e(\pp|\p_\infty)\cdot v_{\p_\infty}(x^{r+1}) = -(r+1)\cdot e(\pp|\p_\infty) < 0,
\]
hence by \cref{eq:hermitian_curve_eq},
\[
v_\pp(y^r+y) = r\cdot v_\pp(y) = -(r+1)\cdot e(\pp|\p_\infty). 
\]
Since $r$ and $r+1$ are coprime, it follows that
\begin{align*}
    e(\pp|\p_\infty) &= r, \\
    v_\pp(x) &= -r, \\
    v_\pp(y) &= -(r + 1).
\end{align*}
In particular, the place $\p_\infty$ is totally ramified. Let $\mathfrak{P}_\infty$ be the only place lying above it.

This function field has $N = r^3 + 1$ rational places - $\mathfrak{P}_\infty$, and another place $\mathfrak{P}_{\alpha, \beta}$ for all $\alpha, \beta \in \FF_q$ such that $\alpha^{r+1} = \beta^r + \beta$.
Taking $x$ as the element with $\deg(x)_\infty=s$, we have $s=r$ and $\sigma = s\frac{q}{N} = \frac{r^3}{r^3 + 1} \leq 1$. 
The genus of this curve is $g = \frac{r(r-1)}{2}$ \cite[Lemma 6.4.4(a)]{Stich}, and hence $\gamma = g \frac{\sqrt{q}}{N} = \frac{r^2(r-1)}{2(r^3 + 1)} = \Theta(1)$. \\
Fix an integer $r \leq T \leq g$. We have
\[
    \rrs(T \mathfrak{P}_\infty) = 
    \spn_{\FF_q} \left\{ x^iy^j: i,j\geq 0,\ ir + j(r+1) \leq T \right\}.
\]
To use \cref{cor: vanishing trace upper bound}, we seek a subspace of $\rrs(T \mathfrak{P}_\infty)$ consisting of functions whose degrees are not divisible by $p$, and that satisfy the condition in \cref{claim_item:independent_S} of \cref{prop:deg(f)=t and independence of S}.
Notice that $r = s$, and that
\[
    \deg(x^i y^j)_\infty = -v_{\mathfrak{P}_\infty} (x^i y^j) = ir + j(r + 1) = (i + j)r + j.
\]
Thus, if $p\nmid j$ then the degree of $x^iy^j$ is not divisible by $p$. Moreover, if $i + j$ is relatively prime to $\ell$, and $j \frac{q}{N} < \frac{1}{3\ell}$, then the monomial $x^iy^j$ satisfies the condition of \cref{claim_item:independent_S}.
Therefore we define $V \leq \rrs(T \mathfrak{P}_\infty)$ by
\begin{equation}\label{def:Hermitian V}
\begin{split}
    V = \spn_{\FF_q} \{ x^iy^j: &i,j\geq 0, \\
    &ir + j(r+1) \leq T, \\
    &j\not\equiv 0 \mod p, \\
    &\gcd(i+j,\,\ell) = 1, \\
    &j < r/(3\ell)\}.
\end{split}
\end{equation}
Every function $f \in V$ must satisfy the conditions of \cref{cor: vanishing trace upper bound}.
We are in a position to obtain a code and calculate its dimension and distance.
\begin{theorem}\label{thm:trace code of Hermitian curve}
    Let $T$ be an integer such that $r\leq T \leq g$, and let $\ell$ be a power of $p$ such that $\FF_p\subseteq\FF_\ell\subseteq\FF_q$. Let $V$ be defined as in \cref{def:Hermitian V}. Let $\mathfrak{P}_1, \ldots, \mathfrak{P}_{N-1}$ be the set of all rational places of the Hermitian function field, except $\mathfrak{P}_\infty$.
    Let $\trcode:V \to \FF_\ell^{r^3}$ be the trace code of $V$ to $\FF_\ell$, defined by
    \begin{equation*}
        \trcode(f) = \left( \tr(f(\mathfrak{P}_1)), \ldots, \tr(f(\mathfrak{P}_{N-1})) \right).
    \end{equation*}
    Then, this code has rate $\rho$ and relative distance $\delta$ satisfying
    \begin{align}
        \rho &= \Omega \left (\frac{T^2 \log(q)}{\ell \log (\ell) q^{5/2}} \right ), \\
        \delta &= 1-\frac{1}{\ell} - O_p \left ( \frac{\sqrt{T}}{\sqrt{q}} + \ell \frac{T}{q} \right ).\label{eq:delta Hermitian curve}
    \end{align}
\end{theorem}

\begin{proof}
    Every function $f \in V$ satisfies the conditions of \cref{cor: vanishing trace upper bound}. Hence, the distance of the code is at least
    \[
        \delta = 1 - \frac{1}{\ell} - O_p(\sqrt{\tau} + \ell \tau),
    \]
    where $\tau\frac{N}{r} = T$.
    Using \cref{def:Hermitian V}, it is easy to verify that
\[
    \dim V = \Omega\left (\frac{1}{\ell}\left(\frac{T}{r}\right)^2  \right ).
\]
    Hence,
    \[
        \rho = \frac{\dim V\cdot \log_\ell(q)}{r^3} = \Omega \left (\frac{T^2 \log(q)}{\ell \log (\ell) q^{5/2}} \right ).
    \]
\end{proof}

We instantiate \cref{thm:trace code of Hermitian curve} to the setting of $\varepsilon$-balanced codes to get an $\FF_2$ linear code.
\begin{theorem}\label{thm:epsilon-balanced codes from Hermitian curve}
    For every $k$ and every $\varepsilon > 0$, there are choices for $T, r=2^a$ such that the Hermitian trace code $\trcode(V)$ from $\FF_q$ to $\FF_2$ is a $[n, \Omega(k)]_2$ linear code that is $\varepsilon$-balanced, with
    \[
        n = O\left ( \frac{k}{\varepsilon^4} \right )^{3/2}.
    \]
\end{theorem}
\begin{proof}
    Let $p = \ell = 2$. Pick $r$ and $T$ with $r \leq T \leq \frac{r(r-1)}{2}$, where $r$ is a power of $2$, and
    \begin{align*}
        T &= \Theta \left ( \frac{k}{\varepsilon^2} \right ), \\
        r &= \Theta \left( \frac{\sqrt{k}}{\varepsilon^2} \right ),
    \end{align*}
    such that the $O$-term in \cref{eq:delta Hermitian curve} is at most $\varepsilon$.
    Consider the construction of \cref{thm:trace code of Hermitian curve} over $\FF_q$ with these choices of $T, r$, and $q = r^2$. It admits an $\varepsilon$-balanced codes over $\FF_2$, with rate
    \[
        \Omega \left ( \frac{T^2 \log(q)}{r^2} \right ) = \Omega(k \log(q)).
    \]
    Since $n = r^3$, we have
    \[
        n = r^3 = O\left (  \frac{k^{3/2}}{\varepsilon^6} \right ) = O\left (  \frac{k}{\varepsilon^4} \right )^{3/2},
    \]
    as claimed.
\end{proof}

\subsection{The Norm-Trace TAG Code}
Here we present some known facts that can be found in \cite{Norm-Trace-epsilon-balanced-codes}.
Let $r$ be a prime power, let $e \geq 2$ be an integer and let $q = r^e$. The extended norm-trace function field is a function field over $\FF_{q}$, with a parameter $u>1$ such that $u \mid \frac{q-1}{r-1}$, and is defined by
\begin{equation}\label{eq:norm-trace definition}
    F = \FF_{q}(x, y), \qquad y^{r^{e-1}} + y^{r^{e-2}} + \cdots + y = x^u.
\end{equation}
The extended norm-trace function field $F/\FF_q$ has genus $g = \frac{(u-1)(r^{e-1} - 1)}{2}$, and 
$$N = r^{e-1}(u(r-1)+1) + 1$$
places of degree one. Hence,
\[
    g \approx \frac{ur^{e-1}}{2}, \qquad N \approx ur^e, \qquad \frac{N}{g} \approx 2r.
\]
Thus we have
\begin{equation}\label{eq:gamma norm-trace}
    \gamma = g \frac{\sqrt{q}}{N} = 
    \Theta\left(r^{e/2 - 1}\right).
\end{equation}
Consider the place $\p_\infty \in \mathbb{P}_{\FF_q(x)}$. Let $\mathfrak{P}$ be some place lying above it. Then
\[
v_\pp(x^u)=e(\pp|\p_\infty)\cdot v_{\p_\infty}(x^u) = -u\cdot e(\pp|\p_\infty) <0 
\]
and hence by \cref{eq:norm-trace definition},
\[
v_\pp(y^{r^{e-1}}+y^{r^{e-2}}+\ldots+y) = r^{e-1}\cdot v_\pp(y) = -u\cdot e(\pp|\p_\infty). 
\]
Since $u \mid q-1$ we have $(u, p) = 1$, hence $r^{e-1} \mid e(\mathfrak{P} | \p_\infty)$. As $e(\pp|\p_\infty)\leq [F:\FF_q(x)] = r^{e-1}$, we conclude that
\begin{align*}
    e(\mathfrak{P}| \p_\infty) &= r^{e-1}, \\
    v_\mathfrak{P}(x) &= -r^{e-1}, \\
    v_\mathfrak{P}(y) &= -u.
\end{align*}
In particular, the place $\p_\infty$ is totally ramified. Let $\mathfrak{P}_\infty$ be the only place lying above it. \\
From now on, assume that $u = \frac{q-1}{r-1} = \sum_{i = 0}^{e-1} r^{i} = O(r^{e-1})$. We remark that in this case, \cref{eq:norm-trace definition} may be written as \[
    F = \FF_q(x, y),\quad \text{Tr}_{\FF_q / \FF_r}(y) = N_{\FF_q / \FF_r}(x).
\]
Taking $x$ as the element with $\deg(x)_\infty=s$, we get that $s = r^{e-1}$ and
\[
    \sigma = s \frac{q}{N} = r^{e-1} \frac{r^e}{ur^e - ur^{e-1} + r^{e-1} + 1} = \frac{r^{e-1}}{u(1-1/r) + 1/r + 1/r^e}.
\]
Since $u = \frac{q-1}{r-1}$, we have $u\left(1-\frac{1}{r}\right) = r^{e-1} - \frac{1}{r}$. Thus we obtain
\[
    \sigma = \frac{r^{e-1}}{r^{e-1} + 1/r^e} \leq 1.
\]
Fix an integer $r^{e-1} \leq T \leq g$. We have
\[
    \rrs(T \mathfrak{P}_\infty) = 
    \spn_{\FF_q} \left\{ x^iy^j: i,j\geq 0,\ ir^{e-1} + ju \leq T \right\}.
\]
To use \cref{cor: vanishing trace upper bound}, we seek a subspace of 
$\rrs(T \mathfrak{P}_\infty)$ consisting of functions whose degrees are not divisible by $p$, and that satisfy the condition in \cref{claim_item:independent_S} of \cref{prop:deg(f)=t and independence of S}.
Notice that
\[
    \deg(x^iy^j)_\infty = -v_{\mathfrak{P}_\infty} (x^i y^j) = ir^{e-1} + ju = 
    (i + j)r^{e-1} + j\sum_{i=0}^{e-2} r^i.
\]
Thus, if $p\nmid j$ then the degree of $x^iy^j$ is not divisible by $p$.
Write $j = a(r-1) + b$ for some integers $a, b$ such that $0\leq a < r^{e-2}$ and $0 \leq b < r-1$. Then,
\begin{align*}
    -v_{\mathfrak{P}_\infty} (x^i y^j) &= (i + j)r^{e-1} + (a(r-1)+b)\sum_{i=0}^{e-2} r^i \\
    &= (i + j)r^{e-1} + a(r^{e-1} - 1) + b\sum_{i=0}^{e-2} r^i \\
    &= (i + j + a)r^{e-1} + b \sum_{i=0}^{e-2} r^i - a.
\end{align*}
Notice that $s=r^{e-1}$. If $\gcd(i + j + a,\,\ell) = 1$, and $\left ( b \sum_{i=0}^{e-2} r^i - a \right ) \frac{q}{N} = O \left ( \frac{b}{r} \right ) < \frac{1}{\ell}$, then the monomial $x^iy^j$ satisfies the condition of \cref{claim_item:independent_S}.
Let $V \leq \rrs(T \mathfrak{P}_\infty)$
be defined by
\begin{equation}\label{def:norm-trace V}
\begin{split}
    V = \spn_{\FF_q} \{ x^iy^j: &i,j\geq 0, \\
    &ir^{e-1} + ju \leq T, \\
    &j \neq 0 \mod p \\
    &\gcd(i + j + a, \ell) = 1, \\
    &b < O \left ( \frac{r}{\ell} \right )\}.
\end{split}
\end{equation}
Thus, every function $f \in V$ satisfies the conditions of \cref{cor: vanishing trace upper bound}.
\begin{theorem}\label{thm:trace code of norm trace curves}
    Let $T$ be an integer such that $r^{e-1} \leq T \leq g$, and let $\ell$ be a power of $p$ such that $\FF_p\subseteq\FF_\ell\subseteq\FF_q$. Let $V$ be defined as in \cref{def:norm-trace V}. Let $\mathfrak{P}_1, \ldots, \mathfrak{P}_{N-1}$ be the set of all rational places in the norm-trace function field, except $\mathfrak{P}_\infty$.
    Let $\trcode:V \to \FF_\ell^{N - 1}$ be the trace code to of $V$ to $\FF_\ell$ , defined by
    \begin{equation*}
        \trcode(f) = \left( \tr(f(\mathfrak{P}_1)), \ldots, \tr(f(\mathfrak{P}_{N-1})) \right).
    \end{equation*}
    Then,
    \begin{align}
        \rho &= \Omega \left (\frac{T^2 \log(q)}{\ell \log (\ell) r^{4e - 3}} \right ), \\
        \delta &= 1-\frac{1}{\ell} - O_p \left ( \sqrt{\frac{T}{r^{e/2 + 1}}} + \ell \frac{T}{r^{e}} \right ).\label{eq:delta norm-trace}
    \end{align}
\end{theorem}
\begin{remark}
    Since $r^{e-1} \leq T$, \cref{eq:delta norm-trace} is non-trivial only if $e < 4$.
\end{remark}

\begin{proof}
    Every function $f \in V$ satisfies the conditions of \cref{cor: vanishing trace upper bound}. Hence, the distance of the code is at least
    \[
        \delta = 1 - \frac{1}{\ell} - O_p(\sqrt{\tau} \gamma + \ell \tau \gamma),
    \]
    where $\tau\frac{N}{r^{e/2}} = T$. We have,
\[
    \tau = \Theta \left ( T \frac{r^{e/2}}{r^{2e-1}} \right ) = \Theta \left ( \frac{T}{r^{3e/2 - 1}} \right ).
\]
Together with \cref{eq:gamma norm-trace}, we obtain \cref{eq:delta norm-trace}.
To analyze the rate of the code,
using \cref{def:norm-trace V} it is easy to verify that
\[
    \dim V = \Omega\left (\frac{1}{\ell}\left(\frac{T}{r^{e-1}}\right)^2  \right ).
\]
    We output $N-1 = \Theta\left(r^{2e-1}\right)$ elements in $\FF_\ell$, hence
    \[
        \rho = \frac{\dim V\cdot\log_\ell(q)}{N} = \Omega \left (\frac{T^2 \log(q)}{\ell \log (\ell) r^{4e - 3}} \right ).
    \]
\end{proof}
We instantiate \cref{thm:trace code of norm trace curves} with $e = 3$ in the setting of $\varepsilon$-balanced codes, thereby obtaining an $\FF_2$-linear code. In this case, the resulting $\varepsilon$-balanced codes are significantly weaker than those in \cref{thm:epsilon-balanced codes from Hermitian curve}, primarily because $\gamma$ is non-constant and appears in the error term of the relative distance.
\begin{theorem}
    For every $k$ and every $\varepsilon > 0$, there are choices for $T, r=2^a$ such that the trace code $\trcode(V)$ of the norm-trace curve with $e=3$ from $\FF_q$ to $\FF_2$ is a $[n, \Omega(k)]_2$ linear code that is $\varepsilon$-balanced, with
    \[
        n = O\left ( \frac{k}{\varepsilon^4} \right )^{5}.
    \]
\end{theorem}
\begin{proof}
    In the above setting, let $e=3$, and let $p = \ell = 2$. Pick $r,T$ such that $r^2 \leq T \leq g$ where $r$ is a power of $2$, and
    \begin{align*}
        T &= \Theta \left ( \frac{k^{5/2}}{\varepsilon^8} \right ), \\
        r &= \Theta \left( \frac{k}{\varepsilon^4} \right ),
    \end{align*}
    such that the $O$-term in \cref{eq:delta norm-trace} is at most $\varepsilon$.
    Consider the construction of \cref{thm:trace code of norm trace curves} over $\FF_q$ with these choices of $e=3,\ T,\ r$, and $q = r^2$. It admits an $\varepsilon$-balanced codes over $\FF_2$, with rate
    \[
        \Omega \left ( \frac{T^2 \log(q)}{r^4} \right ) = \Omega(k \log(q)) = \Omega(k).
    \]
    Since $n = O(r^5)$, we have
    \[
        n = O(r^5) = O\left (  \frac{k}{\varepsilon^4} \right )^{5},
    \]
    as claimed.
\end{proof}

\subsection{The Hermitian Tower TAG Code}
In this subsection we mainly follow \cite[Section 3.1]{Folded_codes_from_function_field_towers}.

Let $p$ be a prime number, and $r = p^\ell$ for some integer $\ell \ge 1$. Let $q = r^2$ and let $e \leq r/2$ be an integer. The Hermitian tower is defined by the following recursive equations
\[
    x_{i+1}^r + x_{i+1} = x_i^{r+1},\qquad i=1,2,\ldots, e-1,
\]
and $F_e = \FF_q(x_1, x_2, \ldots, x_e)$.
The place $\p_\infty \in \mathbb{P}_{\FF_q(x)}$ is totally ramified in $F_e$ and let $\mathfrak{P}_\infty$ be the unique place lying above it. This is a rational place. There are exactly $r^{e+1}$ more rational places in $\mathbb{P}_{F_e}$, coming from $e$-tuples $(\alpha_1,\ldots, \alpha_e) \in \FF_q^{e}$ such that $\alpha_{i+1}^r + \alpha_{i+1} = \alpha_i^{r+1},\ i=1,2,\ldots, e-1.$ The genus of $F_e$ is
\[
    g_e = \frac{1}{2}\left ( \sum_{i=1}^{e-1}r^e\left (1+\frac{1}{r} \right )^{i-1} - (r+1)^{e-1} + 1 \right ) \leq er^e.
\]
Hence,
\begin{equation}\label{eq:gamma Hermitian tower}
    \gamma = g \frac{\sqrt{q}}{N} \leq er^e \frac{r}{r^{e+1}} = e.
\end{equation}
Taking $x$ as the element with $\deg(x)_\infty=s$, we have $s = -v_{\pp_\infty}(x) = e(\pp_\infty|\p_\infty) = r^{e-1}$. Hence 
\[
    \sigma = s \frac{q}{N} = r^{e-1} \frac{r^2}{r^{e+1} + 1} \leq 1.
\]
Moreover, for all $1 \leq i \leq e$, we have
\[
    v_{\mathfrak{P}_\infty}(x_i) = r^{e-i}(r+1)^{i-1}.
\]
Next, fix an integer $r^{e-1} \leq T \leq g_e \leq er^{e}$. We have
\begin{equation*}
    \rrs(T\mathfrak{P}_\infty) = \spn_{\FF_q}\left \{ x_1^{j_1} \cdots x_e^{j_e}: (j_1, \ldots, j_e ) \in \mathbb{Z}^e_{\geq 0},\, \sum_{i=1}^e j_i r^{e-i}(r+1)^{i-1} \leq T \right \}.
\end{equation*}
To use \cref{cor: vanishing trace upper bound}, we seek a subspace of $\rrs(T \mathfrak{P}_\infty)$ consisting of functions whose degrees are not divisible by $p$, and that satisfy the condition in \cref{claim_item:independent_S} of \cref{prop:deg(f)=t and independence of S}.
Notice that $s = r^{e-1}$, and that
\begin{align*}
    -v_{\mathfrak{P}_\infty} (x_1^{j_1} \cdots x_e^{j_e}) &= \sum_{i=1}^e j_i r^{e-i}(r+1)^{i-1} \\
    &= j_1 r^{e-1} + \sum_{i=2}^e j_i (r^{e-1} + O_e(r^{e-2})).
\end{align*}
Hence, if $j_e$ is not divisible by $p$, then $\deg (x_1^{j_1} \cdots x_e^{j_e})_\infty$ is not divisible by $p$. Moreover, if $j_i \leq O_e(\frac{r}{\ell})$ for all $i \geq 2$, then this monomial satisfies the condition of \cref{claim_item:independent_S}.
Let $V \leq \rrs(T\mathfrak{P}_\infty)$ be defined by
\begin{equation}\label{def:Hermitian tower V}
\begin{split}
    V = \spn_{\FF_q} \{x_1^{j_1} \cdots x_e^{j_e}\,:\  &j_i\geq 0, \\
    &\sum_{i=1}^e j_i r^{e-i}(r+1)^{i-1} \leq T, \\
    &j_e \neq 0 \mod p, \\
    &\gcd \left (\sum_{i=1}^e j_i,\,\ell \right ) = 1, \\
    &j_i \leq O_e \left ( \frac{r}{\ell} \right )\forall i \geq 2\}.
\end{split}
\end{equation}
Thus, every function $f \in V$ satisfies the conditions of \cref{cor: vanishing trace upper bound}.
\begin{theorem}\label{thm:trace code of Hermitian tower}
    Let $T$ be an integer such that $r^{e-1} \leq T \leq g_e \leq er^e$, and let $\ell$ be a power of $p$ such that $\FF_p \subseteq \FF_\ell \subseteq \FF_q$.
    Let $V$ be defined as in \cref{def:Hermitian tower V}. Let $\mathfrak{P}_1, \ldots, \mathfrak{P}_{N-1}$ be the set of all rational places in $F_e$, except $\mathfrak{P}_\infty$.
    Let $\trcode:V \to \FF_\ell^{N - 1}$ be the trace code to of $V$ to $\FF_\ell$ , defined by
    \begin{equation*}
        \trcode(f) = \left( \tr(f(\mathfrak{P}_1)), \ldots, \tr(f(\mathfrak{P}_{N-1})) \right).
    \end{equation*}
    Then,
    \begin{align}
        \rho &= \Theta_e \left (\frac{T^e \log(q)}{\ell^{e-1} r^{e(e-1)} \log (\ell)} \right ), \\
        \delta &= 1-\frac{1}{\ell} - O_{p, e} \left ( \sqrt{\frac{T}{r^e}} + \ell \frac{T}{r^{e}} \right ).\label{eq:delta Hermitian tower}
    \end{align}
\end{theorem}

\begin{proof}
    Every function $f \in V$ satisfies the conditions of \cref{cor: vanishing trace upper bound}. Hence, the relative distance of the code is at least
    \[
        \delta = 1 - \frac{1}{\ell} - O_p(\sqrt{\tau} \gamma + \ell \tau \gamma),
    \]
    where $\tau\frac{N}{r} = T$. We have,
\[
    \tau = \Theta \left ( T \frac{r}{r^{e+1}} \right ) = \Theta \left ( \frac{T}{r^{e}} \right ).
\]
Together with \cref{eq:gamma Hermitian tower}, we obtain \cref{eq:delta Hermitian tower}.
    Using \cref{def:Hermitian tower V}, it is easy to verify that
\[
    \dim V = \Omega_e \left ( \frac{T}{r^{e-1}} \left ( \frac{T}{\ell r^{e-1}} \right )^{e-1}  \right )
\]
    Since we output $N - 1 = r^{e+1}$ elements in $\FF_\ell$, we conclude
    \[
        \rho = \frac{\dim V\cdot \log_\ell(q)}{N} = \Theta_e \left (\frac{T^e \log(q)}{\ell^{e-1} r^{e(e-1)} \log (\ell)} \right ).
    \]
\end{proof}
We instantiate \cref{thm:trace code of Hermitian tower} to the setting of $\varepsilon$-balanced codes to get an $\FF_2$ linear code.

\begin{theorem}\label{thm:epsilon-balanced codes from Hermitian tower}
    For every $k$, $\varepsilon > 0$ and $e \geq 2$, there are choices for $T, r=2^a$ such that the trace code of the Hermitian tower of level $e$ $\trcode(V)$ from $\FF_q$ to $\FF_2$ is a $[n, \Omega_e(k)]_2$ linear code that is $\varepsilon$-balanced, with
    \[
        n = O_e\left ( \frac{k}{\varepsilon^{2e}} \right )^{\frac{e+1}{e}}.
    \]
\end{theorem}
\begin{proof}
    Let $p = \ell = 2$. Pick $r, T$ such that $r^{e-1} \leq T \leq g_e \leq er^e$, where $r$ is a power of $2$, and
    \begin{align*}
        T &= \Theta_e \left ( \frac{k}{\varepsilon^{2(e-1)}} \right ), \\
        r &= \Theta_e \left( \frac{k^{1/e}}{\varepsilon^{2}} \right ),
    \end{align*}
    such that the $O$-term in \cref{eq:delta Hermitian tower} is at most $\varepsilon$.
    Consider the construction of \cref{thm:trace code of Hermitian tower} over $\FF_q$ with these choices of $T, r$, and $q = r^2$. It admits an $\varepsilon$-balanced codes over $\FF_2$, with rate
    \[
        \Omega_e \left ( \frac{T}{r^{e-1}} \right )^{e} = \Omega_e \left ( \frac{k}{\varepsilon^{2(e-1)}} \frac{\varepsilon^{2(e-1)}}{k^{\frac{e-1}{e}}} \right )^e = \Omega_e(k).
    \]
    Since $n = r^{e+1}$, we have
    \[
        n = r^{e+1} = O_e\left (  \frac{k^{\frac{e+1}{e}}}{\varepsilon^{2(e+1)}} \right ) = O_e   \left (  \frac{k}{\varepsilon^{2e}} \right )^{\frac{e+1}{e}},
    \]
    as claimed.
\end{proof}

\section{TAG Codes vs.\ Concatenation in the High Distance Regime}\label{sec:trace_vs_had}

In this section we compare binary TAG codes with the familiar approach of concatenating with the Hadamard code, thereby proving \cref{prop:had vs tag}. We consider here the case where the error term guaranteed in \cref{cor: vanishing trace upper bound} is $O(\tau)$, which would lead to the $\Omega(\eps^3)$ lower bound on the rate. The ``Moreover'' part of \cref{prop:had vs tag}, corresponds to $O(\sqrt{\tau})$, is similar.

Let $C$ be a one-point evaluation AG code, defined by an algebraic function field $F / \FF_q$ with genus $g$, and a set of $N$ rational places $Y$ with respect to a divisor $G = T \mathfrak{P}$ such that $\mathfrak{P} \not\in Y$, $T < g$.
Assume that $F$ contains an element $x \in F$ with $\frac{q}{N}\deg (x)_\infty  = 1 + O(q^{-1/2})$, and $N / g = \Omega(\sqrt{q})$. Let $\ell = \ell(T \mathfrak{P})$, the dimension of the corresponding Riemann-Roch space. This AG code has parameters $n_1 = N$, designated distance $d_1 = N - T$, and dimension $\ell$, and it is defined over $\FF_q$.
We compare the parameters of the two methods described above:

\begin{enumerate}
    \item \textbf{TAG codes.} Without loss of generality, assume that every function $f \in \rrs(T\mathfrak{P})$ satisfies the conditions of \cref{cor: vanishing trace upper bound}. The error term guaranteed in \cref{cor: vanishing trace upper bound} is $O(\tau) = O\left (\frac{t \sqrt{q}}{N} \right )$. Then the resulting TAG code has parameters
    \begin{align*}
        n_T &= N,\\
        k_T &= \ell \log (q),\\
        \varepsilon_T &= \Omega\left(\frac{T \sqrt{q}}{N}\right).
    \end{align*}
    \item \textbf{Concatenation with Hadamard.} The resulting code has parameters
    \begin{align*}
        n_H &= N q,\\
        k_H &= \ell \log (q),\\
        \varepsilon_H &\le \frac{T}{N}.
    \end{align*}
\end{enumerate}
We let $k = k_H = k_T$.
Ta-Shma and Ben-Aroya~\cite{Ben-Aroya_Ta-Shma} proved that in the regime of $g = \Omega(\sqrt{q})$ we have
\begin{equation}\label{eq:amnon}
n_H = \Omega\left( \frac{k}{\varepsilon_H^{2.5}\log(k/\varepsilon_H)} \right ).
\end{equation}
Assume that $\alpha,\beta > 0$ are such that
\begin{equation}\label{eq:nh}
n_H = \Omega\left(\frac{k^\alpha}{\varepsilon_H^{\beta}}\right),
\end{equation}
and note that per \cref{eq:amnon}, $\beta \geq 2.5$. Observe that $\varepsilon_T = \Omega(\varepsilon_H \sqrt{q})$, $n_T = n_H / q$. Hence,
\[
    n_T = \frac{n_H}{q}= \Omega\left(\frac{k^{\alpha} \varepsilon_H^{-\beta}}{q}\right) = \Omega\left(k^{\alpha} \varepsilon_T^{-\beta} q^{\beta / 2 - 1}\right).
\]
Notice that by \cref{clm:LB_for_the_gonality},
in order to have $\rrs(T\mathfrak{P}) \neq \FF_q$, we must have $T \geq \frac{N}{q+1}$, hence as $\eps_T = \Omega( \frac{T \sqrt{q}}{N})$ we have $\sqrt{q} = \Omega(\frac{1}{\varepsilon_T})$. Finally, we obtain
\[
    n_T = \Omega\left(k^{\alpha} \varepsilon_T^{-\beta - (\beta - 2)}\right) = \Omega\left(\frac{k^{\alpha}}{ \varepsilon_T^{2\beta-2}}\right).
\]
Comparing this with \cref{eq:nh}, recalling that $\beta \ge 2.5$, wee see that the exponent in the dependence in $\eps$ in the TAG code, $2\beta-2$, is larger by at least $0.5$ from the corresponding exponent for the concaentated code. Indeed, the difference is $(2\beta-2)-\beta = \beta -2 \ge 0.5$.

\section*{Acknowledgment}
We are grateful to Amnon Ta\hbox{-}Shma for illuminating discussions and for sharing his insights on trace codes of algebraic geometric codes.

\bibliographystyle{alpha}

\end{document}